\documentclass[runningheads]{llncs}
\usepackage{graphicx}
\usepackage{amsmath}
\usepackage{amssymb}
\usepackage{mathtools}
\usepackage{comment}
\usepackage{subfig}
\usepackage[ruled,vlined]{algorithm2e}
\usepackage{graphics}
\usepackage{wrapfig,booktabs}
\usepackage{tabularx}
\usepackage[utf8]{inputenc}

\begin{document}
	\title{Reducing Boolean Networks with\\ Backward Boolean Equivalence - extended version\thanks{Partially supported by the %Independent Research Fund Denmark 
			DFF project REDUCTO 9040-00224B, the Poul Due Jensen Foundation grant 883901, and the PRIN project SEDUCE 2017TWRCNB.
	}}
	\titlerunning{Reducing Boolean Networks with Backward Boolean Equivalence}
	\author{
		Georgios Argyris\inst{1} \orcidID{0000-0002-3203-0410} \and
		Alberto Lluch Lafuente \inst{1} \orcidID{0000-0001-7405-0818}\and 
		Mirco Tribastone\inst{2}\orcidID{0000-0002-6018-5989}\and
		Max Tschaikowski\inst{3}\orcidID{0000-0002-6186-8669}\and
		Andrea Vandin\inst{4,1}\orcidID{0000-0002-2606-7241}
	}
	\authorrunning{G. Argyris, A. Lluch Lafuente, M. Tribastone, M. Tschaikowski, A. Vandin}
	\institute{DTU Technical University of Denmark, Kongens Lyngby, Denmark\and
		IMT School for Advanced Studies Lucca, Italy\and
		University of Aalborg, Denmark\and
		Sant'Anna School for Advanced Studies, Pisa, Italy
	}
	\maketitle              % typeset the header of the contribution
	\begin{abstract}
		Boolean Networks (BNs) are established models to qualitatively describe  biological systems. The analysis of BNs might be infeasible for medium to large BNs due to the state-space explosion problem. We propose a novel reduction technique called \emph{Backward Boolean Equivalence} (BBE), which preserves some properties of interest of BNs. In particular, reduced BNs provide a compact representation by grouping variables that, if initialized equally, are always updated equally. The resulting reduced state space is a subset of the original one, restricted to identical initialization of grouped variables. The corresponding trajectories of the original BN can be exactly restored. We show the effectiveness of BBE by performing a large-scale validation on the whole GINsim BN repository. In selected cases, we show how our method enables analyses that would be otherwise intractable. Our method complements, and can be combined with, other reduction methods found in the literature.
	\end{abstract}
	\section{Introduction}
	\setcounter{footnote}{0}

	Boolean Networks (BNs)
	are an established method to model biological systems~\cite{KAUFFMAN1969437}. A BN consists of Boolean variables (%in some studies
	also called nodes) which represent the activation status of the components in the model.
	%, e.g. species, genes, molecules, or proteins.
	%which represent the spevariables\comg{it will be good to mention somewhere that the terms "BN nodes" and "BN variables" is the same thing and will be used interchangeably},
	The variables are commonly depicted as nodes in a network with directed links which represent influences between them.
	However, a full descriptive mathematical model underlying a BN consists of a set of Boolean functions, the \emph{update functions}, that govern the Boolean values of the variables. Two BNs are displayed on top of Fig.~\ref{fig:image1}. The BN on the left has three variables $x_1$, $x_2$, and $x_3$, and the BN on the right has two variables $x_{1,2}$ and $x_3$. %The number of variables constitute the main dimensionality of a BN.
	%The dimensionality of a BN is given by its number of nodes.
	%Since the maximum number of links depends on the number of nodes, it is common to focus on the latter as main dimension.
	The dynamics (the state space) of a BN is encoded into a \emph{state transition graph}~(STG). The bottom part of Fig.~\ref{fig:image1} displays the STGs of the corresponding BNs.
	The boxes of the STG represent the BN \emph{states}, i.e.  vectors with one Boolean value per BN variable.
	%The states of an STG, represented as boxes in the STGs of Fig. \ref{fig:image1}, are vectors with one Boolean value per BN node.
	A directed edge among two STG states represents the evolution of the system from the source state to the target one. The target state is obtained by synchronously applying  all the update functions to the activation values of the source state.
	There exist BN variants with other update schema, e.g.  asynchronous non-deterministic~\cite{thomas2013kinetic} or probabilistic~\cite{shmulevich2002probabilistic}. Here we focus on the synchronous case. 
	%Typically, a BN is provided with one initial state, from which all reachable STG states are generated by applying iteratively the update functions.
	BNs where variables are \emph{multivalued}, i.e. can take more than two values to express different levels of activation~\cite{thomas1991regulatory}, are supported via the use of  \emph{booleanization} techniques~\cite{delaplace2020bisimilar}, at the cost, however, of increasing the number of variables.
	
	\begin{figure}[ht]
		\centering
		\begin{tabular}{ccc}
			%\centering
			\scalebox{0.85}{
				$
				\begin{array}{rcl}
					x_1(t+1) &=& \neg x_3(t) \vee x_1(t)\\
					x_2(t+1) &=& x_1(t) \vee x_2(t) \vee \neg x_3(t)\\
					x_3(t+1) &=& x_2(t) \wedge \neg x_3(t)
				\end{array}
				$
			}
			\quad
			& %\scriptsize
			$\xRightarrow[\text{$x_1, x_2: BBE$}]{}$
			&
			\scalebox{0.85}{
				$
				\begin{array}{rcl}
					x_{1,2}(t+1) &=& \neg x_3(t) \vee x_{1,2}(t)\\
					x_3(t+1) &=& x_{1,2}(t) \wedge \neg x_3(t)\\
					\phantom{x}
				\end{array}
				$
			}
			%  \\ \\
			% $\Downarrow$ & & $\Downarrow$
			\\
			\includegraphics[width=0.29\linewidth]{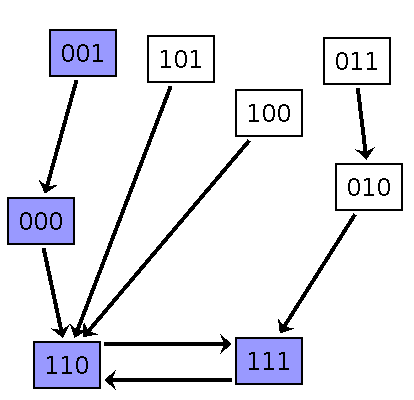}
			&
			&
			\includegraphics[width=0.18\linewidth]{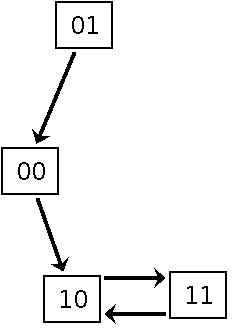}
		\end{tabular}
		\vspace{-0.2cm}
		\caption{A BN (top-left), its STG (bottom-left), the BBE-reduced BN (top-right) and its (reduced) STG (bottom-right).
			%We imported the BNs as .bnet files and generated the STGs with GINsim \cite{naldi2009logical}.
		}
		\label{fig:image1}
	\end{figure}

	BNs suffer from the state space explosion problem: there are exponentially many STG states with respect to the number of BN variables.
	This hampers BN analysis in practice, calling for reduction techniques for BNs. There exist  manual or semi-automated ones %techniques for reducing BNs 
	based on domain knowledge. Such empirical reductions have several drawbacks: being semi-automated, they are error-prone, and do not scale. %and they do not scale with the BN size. 
	%Furthermore, no formal relation exists among the STGs of the original and the reduced BN. 
	Popular examples are those based on the idea of 
	\emph{variable absorption}, % a manual reduction method which was 
	proposed originally in \cite{naldi2011dynamically,veliz2011reduction,saadatpour2013reduction}. The main idea is that certain BN variables can get \emph{absorbed} by the update functions of their target variables %. This is implemented 
	by replacing all occurrences of the absorbed variables with their update functions. Other methods  automatically remove \emph{leaf} variables (variables with 0 outgoing links) or \emph{frozen} variables (variables that stabilize after some iterations independently of the initial conditions)~\cite{richardson2005simplifying,bilke2001stability}. %, as they are dynamically  irrelevant~\cite{bilke2001stability,socolar2003scaling} with respect to the asymptotic behaviour of the BN. 
	Several techniques \cite{figueiredo2016relating,berenguier2013dynamical} focus on reducing the STGs rather than the BN generating them. This requires to construct the original STG, thus still incurring the state space explosion problem.
	
	Our research contributes a novel mathematically grounded method to automatically minimize BNs while exactly preserving behaviors of interest. We present Backward Boolean Equivalence (BBE), which  %minimizes BNs by collapsing 
	collapses \emph{backward  Boolean equivalent} variables.
	%In typical ODEs or CRNs, variables are defined as backward equivalent if their values are equal at all time steps. Hence, backward equivalence requires same initial conditions for backward equivalent nodes.
	The main intuition is that two BN variables are BBE-equivalent if %and only if 
	they maintain equal value in any state reachable from a state wherein they have the same value.
	%they have same activation status in any state of the STG, whenever initialized with same activation status.
	In the STG in Fig.~\ref{fig:image1} (left), we note that for all %STG 
	states where $\mathit{x_1}$ and $\mathit{x_2}$ have same value (purple boxes), the update functions do not distinguish them. %The crucial aspect of 
	Notably, BBE is that it can be checked directly on the BN, without requiring to generate the STG. Indeed, as depicted in the middle of Fig.~\ref{fig:image1} , $\mathit{x_1}$ and $\mathit{x_2}$ can be shown to be BBE-equivalent by inspecting their update functions: If $x_1,x_2$ have the same value in a state, i.e. $\mathit{x_1(t)=x_2(t)}$, then their update functions will not differentiate them since $x_2(t+1) = x_1(t) \vee x_2(t) \vee \neg x_3(t) =x_1(t) \vee x_1(t) \vee \neg x_3(t)=x_1(t) \vee \neg x_3(t)=x_1(t+1)$.
	We also present an iterative partition refinement algorithm \cite{paige1987three} that computes the largest BBE of a BN. 
	%
	%This allows identifying sets of nodes which, when activated equally, change their activation status in the same way in any subsequent state.\coma{previous sentence is a repetition}\comg{in previous sentence we talk about two variables. Instead, we can directly talk about sets of variables and erase the last sentence.}
	%, and to focus the analysis on the dynamics of such behaviours.
	%
	%It is indispensable to specify initial partition consistent with the initial conditions of the original model.
	%
	Furthermore, given a BBE, we obtain a \emph{BBE-reduced} BN by collapsing all BBE-equivalent variables into one in the reduced BN. In Fig.~\ref{fig:image1}, we collapsed $\mathit{x_1,x_2}$ into $\mathit{x_{1,2}}$.
	%
	%The original and the reduced BN are related as the latter preserves the dynamics of the former in a formal sense.
	%
	The reduced BN faithfully preserves part of the dynamics of the original BN: it exactly preserves all states and paths of the original STG where BBE-equivalent variables have same activation status. Fig.~\ref{fig:image1} (right) shows the obtained BBE-reduced BN and its STG. We can see that the purple states of the original STG are preserved in the one of the reduced BN.	
	
	We implemented BBE in ERODE~\cite{cardelli2017erode}, a freely available tool for reducing biological systems. 
	We built %establish a novel 
	a toolchain that combines ERODE with several %state of the art softwares 
	tools for the analysis, visualization and reduction of BNs, allowing us to apply BBE to all BNs from the GINsim repository~(\url{http://ginsim.org/models\_repository}). % to assess its efficiency. 
	BBE led to reduction in  61 out of 85 considered models ($70\%$), facilitating STG generation. For two models, we could obtain the STG of the reduced BN while it is not possible to generate the original STG due to its size. %The algorithm that computes the reduced BN, is initialized with a partition of the BN variables defined by the modeller according to restrictions and desires. We exploit this feature by defining initial partitions such that the reduced BN preserves the attractors.
	We further demonstrate the effectiveness of BBE in three case studies, focusing on their \emph{asymptotic dynamics} by means of \emph{attractors analysis}.  Using BBE, we can identify the attractors of large BNs which would be otherwise intractable. % due to high complexity.
	%By establishing connections between the original and the reduced STG, we computed all attractors of 1 BN whose attractor identification was previously intractable. In the other two cases, we identified all attractors in less time.
	
	The article is organized as follows: Section~\ref{sec2} provides the basic definitions and the running example based on which we will explain the key concepts. In Section~\ref{sec3}~, we introduce BBE, present the algorithm for the automatic computation of maximal BBEs, and formalize how the STGs of the original and the reduced BN are related. In Section~\ref{sec4}, we apply BBE to BNs from the literature. In Section~\ref{sec:related} we discuss related works, while Section~\ref{sec6} concludes the paper. %summarizes our analysis and motivates the reader to future challenges.

	\section{Preliminaries}\label{sec2}
	
	%Various forms of Boolean networks (BNs) have been studied in the literature. \coma{This refers to Boolean/multivalued, synchronous/asynchronous, should be a sentence by itself}.
	
	BNs can be represented visually using some graphical representation which, however, might not contain all the information about their dynamics~\cite{klamt2009hypergraphs}. An example is that of signed interaction (or regulatory) graphs adopted by the tool GinSim~\cite{naldi2009logical}.
	These representations are often paired with a more precise description containing
	%usually provides only a sketch of the model, without precise information about the update functions. There exist representations that capture all the information of a BN. These representations consist of a set of variables tied
	either truth tables \cite{richardson2005simplifying} or algebraic update functions \cite{su2020sequential}. In this paper we focus on such precise representation, and in particular on the latter. % as illustrated in Fig. \ref{fig:image2}.
	However, in order to better guide the reader in the case studies, wherein we manipulate BNs with a very large number of components, we also introduce  signed interaction graphs.
	
	We explain the concepts of current and next sections using the simple BN of Fig.~\ref{fig:image2} (left) taken from \cite{10.1371/journal.pcbi.1000936}. The model refers to the development of the outer part of the brain: the cerebral cortex. This part of the brain contains different areas with specialised functions. The BN is composed of five variables which represent the gradients that take part in its development: the morphogen \emph{Fgf8} and four transcription factors, i.e., \emph{Emx2, Pax6, Coup\_tfi, Sp8}. During development, these genes are expressed in different concentrations across the surface of the cortex forming the different areas. %\coma{REMOVE:}Particularly, the complementary expression gradients of \emph{Emx2, Pax6} with \emph{Couptfi, Sp8} provide an ideal system for modelling the crude simplification of two different areas in late embryonic stages. %The two different areas correspond to the two attractors (steady states) of the BN, which we will explain later on.
	
	\begin{figure}[t]
		\centering
		\scalebox{0.85}{
			$\begin{array}{rcl}
				\mathit{x_{Fgf8}(t\!+\!1)} &=& \mathit{x_{Fgf8}(t) \land \neg x_{Emx2}(t) \land x_{Sp8}(t)}\\
				\mathit{x_{Pax6}(t\!+\!1)} &=& \mathit{\neg x_{Emx2}(t) \land x_{Sp8}(t) \land \neg x_{Coup\_tfi}(t)}\\	
				\mathit{x_{Emx2}(t\!+\!1)} &=& \mathit{\neg x_{Fgf8}(t) \land \neg x_{Pax6}(t) \land \neg x_{Sp8}(t) \land x_{Coup\_tfi}(t)}\\	
				\mathit{x_{Sp8}(t\!+\!1)} &=& \mathit{x_{Fgf8}(t) \land \neg x_{Emx2}(t)}\\	
				\mathit{x_{Coup\_tfi}(t\!+\!1)} &=& \mathit{\neg x_{Fgf8}(t) \land \neg x_{Sp8}(t)}\\
				\\
				\\
				\\
				\\
				\\
				\\
				\\
				\\
				\\
			\end{array}$
		}
		\hspace{-0.45cm}
		\vspace{-2.25cm}
		\includegraphics[width=0.37\linewidth]{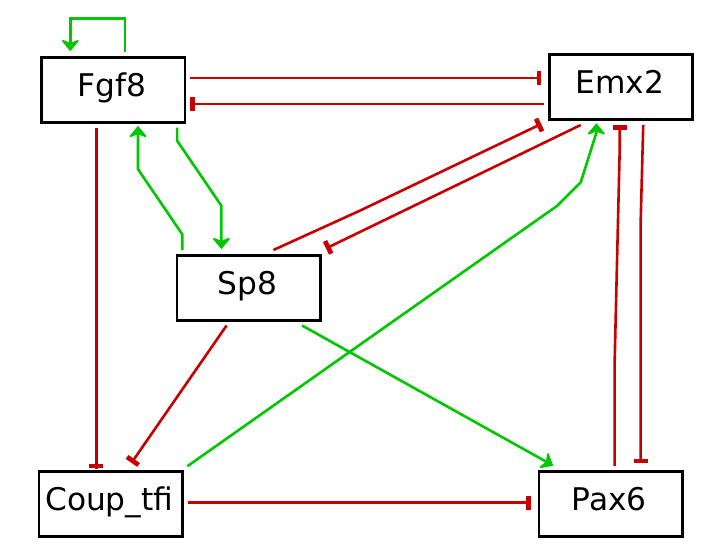}
		\caption{(Left) the BN of cortical area development from~\cite{10.1371/journal.pcbi.1000936}; (Right) its signed interaction graph.}
		\label{fig:image2}
	\end{figure}
	
	%Fig.~\ref{fig:CAD}
	Fig.~\ref{fig:image2} (right)
	displays the signed interaction graph that corresponds to the BN. % of Fig.~\ref{fig:image2}.
	The green arrows correspond to \emph{activations} whereas the red arrows correspond to \emph{inhibitions}.
	For example, the green arrow from $\mathit{Sp8}$ to $\mathit{Pax6}$ denotes that the former promotes the latter because variable $\mathit{x_{Sp8}}$ appears (without negation) in the update function of $\mathit{x_{Pax6}}$, whereas the red arrow from $\mathit{Pax6}$ to $\mathit{Emx2}$ denotes that the former inhibits the latter because the negation of $\mathit{x_{Pax6}}$ appears in the update function of $\mathit{x_{Emx2}}$.
	
	%\begin{figure}[t]
	%	\centering
	%	\includegraphics[width=0.70\linewidth]{Rexample.png}
	%	\caption{The signed intercation graph of the cortical area development BN %of Fig. \ref{fig:image4}
	%		.}
	%	\label{fig:CAD}
	%\end{figure}
	
	We now give the formal definition of a BN:
	
	\begin{definition}%\textbf{(Boolean Network)}
		\label{def:bn}
		A BN is a pair $(X, F)$
		%$(X, F)$
		where $X=\{x_1,...,x_n\}$ is a set of variables
		%denoting the state of the BN \coma{no: denoting the species/agents/main actors of the BN}
		and $F=\{f_{x_1},...,f_{x_n}\}$ is a set of update functions, with $f_{x_i}:\mathbb{B}^n \rightarrow \mathbb{B}$ being the update function of variable $x_i$.
	\end{definition}
	
	A BN is often denoted as $X(t+1)=F(X,t)$, or just
	$X=F(X)$.
	In Fig.~\ref{fig:image2} we have $\mathit{X=\{ x_{Fgf8}, x_{Pax6},}$ $\mathit{x_{Emx2}, x_{Sp8}, x_{Coup\_tfi}\}}$. % is the set of variables,
	%and $F$ consists of the Boolean expressions appearing in the right-hand-side of each equation.
	%$\mathit{ F=\{ f_{x_{Fgf8}}, f_{x_{Pax6}},}$ $\mathit{f_{x_{Emx2}}, f_{x_{Sp8}}, f_{x_{Couptfi}}\}}$.  is the set of update functions with $\mathit{f_{x_{i}}=x_i(t+1)}$.

	% with $\mathit{f_{Fgf8} = Fgf8 \land \neg Emx2 \land Sp8}$ etc.
	
	The \emph{state} of a BN is an evaluation of the variables, %in set $\mathit{X}$
	denoted with the vector of values $\mathit{\textbf{s}=(s_{x_1}, \ldots, s_{x_n}) \in \mathbb{B}^n}$.  The variable $\mathit{x_i}$ has the value $\mathit{s_{x_i}}$.
	%Various update schemas have been proposed in literature: like synchronous, sequential, random-asynchronous, random-asynchronous with priority classes etc. For a nice classification of BNs according to their updating schema, someone can read \cite{gershenson2002classification}.
	When the update functions are applied synchronously, we have synchronous transitions between states, i.e. %a pair of states $\mathit{(\textbf{s},\textbf{t})}$ (also denoted $\textbf{s}\xrightarrow{}\textbf{t}$)
	for $\textbf{s},\textbf{t}\in\mathbb{B}^n$ we have $\textbf{s}\xrightarrow{}\textbf{t}$
	if $\mathit{\textbf{t}=F(\textbf{s})=(f_{x_1}(\textbf{s}), \ldots, f_{x_n}(\textbf{s}))}$.
	
	Suppose that the activation status of the variables  $\mathit{x_{Fgf8}}$, $\mathit{x_{Emx2}}$, $\mathit{x_{Pax6}}$, $\mathit{x_{Sp8}}$, $\mathit{x_{Coup\_tfi}}$
	is given by the state $\textbf{s}=(1,0,1,1,1)$. After applying the update functions, we have $\textbf{t}=F(\textbf{s})=(0,0,0,0,0)$.
	
	The state space of a BN, called \emph{State Transition Graph (STG)}, is the set of all possible states and state transitions. %, formally defined as follows.

	\begin{definition}%\textbf{(Synchronous State Transition Graph)}
		Let $B = (X,F)$ be a BN. We define the state transition graph of $B$, denoted with $STG(B)$, as a pair $(S,T)$ with
		%$S = \mathbb{B}^{n}$
		$S \subseteq \mathbb{B}^{n}$
		being a set of vertices labelled with the states of $B$, and $T=\{ \textbf{s}\xrightarrow{}\textbf{t} \mid \textbf{s}\in S, \textbf{t} = F(\textbf{s}) \}$ a set of directed edges representing the transitions between states of $B$.
	\end{definition}
	
	We often use the notation %$\textbf{s}\xrightarrow{}\textbf{t}$ for transitions $(\textbf{s},\textbf{t})$, and
	$\textbf{s} \xrightarrow{}^+\textbf{t}$ for the transitive closure of the transition relation. %\comal{To be discussed if we really want to have these redundant notation.}
	The cardinality of the set of states is $2^n$, which illustrates the state space explosion: we have exponentially many states on BN variables.
	Fig.~\ref{fig:image3}(a) displays the STG of the BN in Fig.~\ref{fig:image2}.
	%
	%The boxes represent BN states which correspond to evaluations of the variables $\mathit{ x_{Fgf8}, x_{Emx2}, x_{Pax6}, x_{Sp8}}$, $x_{\mathit{Couptfi}}$ respectively. The arrows represent the transitions between states.
	%
	
	\begin{comment}
	\begin{figure}
	\centering
	\includegraphics[scale=0.5]{STGCAD.pdf}
	\caption{The STG that of the cortical area development BN of Fig.~\ref{fig:image2}. We use GINsim's visual representation, where self-loops are implicit in nodes without outgoing edges.}
	\label{fig:image3}
	\end{figure}
	\end{comment}

	\begin{figure}[t]
		%	\centering
		\subfloat[STG of BN in Fig.~\ref{fig:image2}]{
			\includegraphics[scale=0.44]{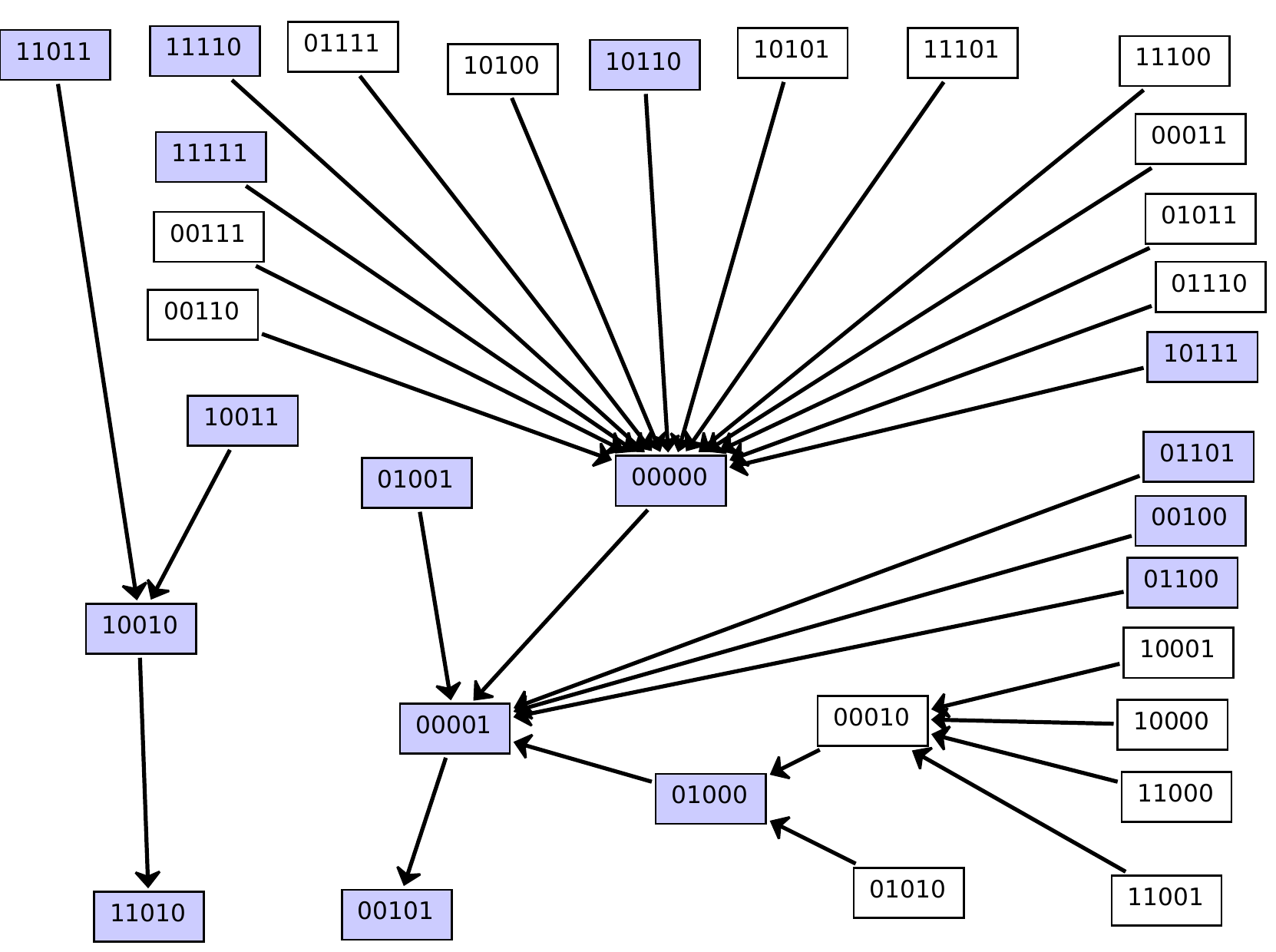}
		}
		%\hspace{.25cm}
		%\vline
		%\vspace{-0.25cm}
		\qquad
		\subfloat[STG of %BBE-reduced
		BN in Fig.~\ref{fig:reducedCAD}]{
			\includegraphics[scale=0.44]{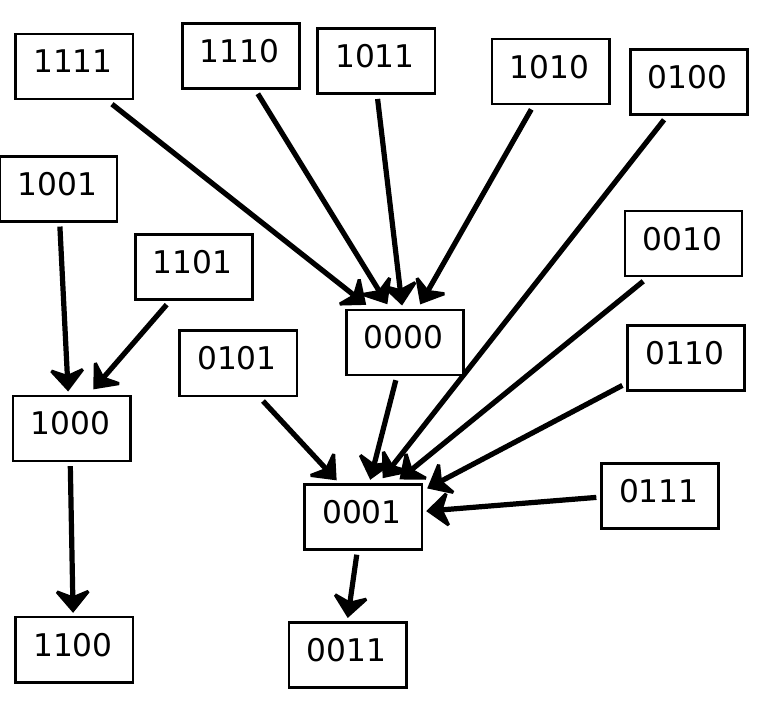}
		}
		\caption{The STGs of the % that of the original cortical area development
			BN of Fig.~\ref{fig:image2} and of its BBE-reduction in. %. Right: The STG that of the reduced cortical area development BN of
			Fig.~\ref{fig:reducedCAD}.
			We use GINsim's visual representation, where self-loops are implicit in nodes without outgoing edges.}
		\label{fig:image3}
	\end{figure}
	
	Several %emergent 
	BN properties are identified in STGs, e.g. attractors, basins of attraction, and transient trajectories~\cite{SCHWAB2020571}. % \coma{??? Is it a definition for attractors, or are you talking about a special class of attractors?}.
	Attractors are sets of states towards which a system tends to evolve and remain~\cite{hopfensitz2013attractors}. They are often associated with the interpretation of the underlying system; for example, Kauffman equated attractors with different cell types \cite{drossel2008random}.
	%In the cortical area development BN, the two attractors model the \emph{crude simplification of only two different areas in late embryonic stages}~\cite{10.1371/journal.pcbi.1000936}\coma{What does it mean?}.
	%Since each cell contains the same DNA (i.e., the same network\coma{what is a network?}\comg{same kind of variables? or same genes?}), cells can only differ by the pattern of gene activity.
	Hence, the main reduction methods that have been developed in the literature so far concentrate on how they affect the asymptotic dynamics i.e. the number of attractors and the distribution of their lengths. We define an attractor as follows: %In this manuscript, we formally define attractors in the context of graph theory.
	
	\begin{definition}\textbf{(Attractor)}\label{att}
		Let $B = (X,F)$ be a BN with $STG(B) = (S,T)$. We say that a set of states $A \subseteq S$ is an \emph{attractor} iff
		\begin{enumerate}
			\item $\forall \textbf{s},\textbf{s}'\in A$, $\textbf{s}\xrightarrow{}^+ \textbf{s}'$, and
			\item $\forall \textbf{s}\in A, \forall \textbf{s}'\in S$, $\textbf{s}\xrightarrow{}^+ \textbf{s}'$ implies $\textbf{s}'\in A$.
		\end{enumerate}
		%= \{\textbf{s}'\mid s' \in A \text{ or } \exists \textbf{s}\in A, \textbf{s}\xrightarrow{}^+ \textbf{s}'\}$. % implies  $\textbf{s}'\in A$.
	\end{definition}
	
	Attractors are hence just absorbing strongly connected components in the STG. An attractor $A$ such that $|A|=1$ is called a \emph{steady state} (also named \emph{point attractor}). We also denote with $|A|$ the \emph{length} of attractor $\mathit{A}$.

	\section{Backward Boolean Equivalence}\label{sec3}
	Our reduction method is based on the notion of backward equivalence, recast for BNs, which proved to be effective for reducing the dimensionality of
	%biological systems given as
	ordinary differential equations \cite{DBLP:conf/popl/CardelliTTV16,DBLP:journals/tcs/CardelliTTV19a} and chemical reaction networks \cite{cardelli2017maximal,DBLP:conf/concur/CardelliTTV15,DBLP:conf/tacas/CardelliTTV16}.
	Section~\ref{sec:bbe} introduces \emph{Backward Boolean Equivalence} (BBE), which is an equivalence relation on the variables of a BN, and use it to obtain a reduced BN. % (\emph{BN reduction}).
	Section~\ref{sec:compute} %provide a logical characterization of BBE, %to check if a given partition %$\mathit{X_R}$
	%is a BBE,
	%and
	provides an algorithm which %exploits the BBE-definition to 
	iteratively compute the maximal BBE of a BN.
	Section~\ref{sec3.3} relates the properties of an original and BBE-reduced BN.
	%\comal{TODO: fix the above since it is confusing that we enumerate 2 things, then we list 2 more. Plus, the mapping to subsections is not explicit and not obvious. This is a long section, so we need to help the reader navigate in it.}
	%
	%Finally, we establish connections between the STGs of the original and the reduced BN in terms of morphisms\coma{do we really do this?/do we want to say this?}.
	
	%In this section we 
	We fix a BN $B=(X,F)$, with $|X|=n$. % unless stated otherwise. 
	We use $R$ to denote equivalence relations on $X$ and $X_R$ for the induced partition.
	
	\subsection{Backward Boolean Equivalence and BN Reduction}
	%Every equivalence relation $R$ induces a partition $H$ on the set $X$ of variables and every partition $H$ corresponds to an equivalence relation $R$. Hence, we will call BBE both the equivalence relation and its corresponding partition.
	%\coma{$R$ and $X$ do not seem to be related. Should we use $X/R$ or $H_R$ rather than $H$?}
	
	\label{sec:bbe}
	\begin{comment}
	x1(t+1) = x3(t)
	x2(t+1) = x3(t)
	x3(t+1) = x3(t)
	x4(t+1) = x4(t)
	x1 BBE x2 according to definition 2.1
	
	x1(t+1) = x3(t)
	x2(t+1) = x4(t)
	x3(t+1) = x3(t)
	x4(t+1) = x4(t)
	x1 NOT BBE x2 according to definition 2.1
	
	But {x1,x2},{x3,x4} is a BBE. Therefore 2.1 has to be made weaker
	
	Our model
	x1(t+1)=x3(t)
	x2(t+1)=x4(t)
	x3(t+1)=x3(t)
	x4(t+1)=x4(t)
	
	x1 is not equivalent to x2 according to definition 2.1
	
	x1(0)=1
	x2(0)=1
	x3(0)=0
	x4(0)=0
	R1: $\{x1,x2\},\{x3,x4\}$=(x1,x2),(x2,x1)(x3,x4),(x4,x3)
	R1 is a BBE (def 2.1)
	R1 is not a BBE (def 2.3)
	(x1,x2) are not eq
	
	x1(0)=1
	x2(0)=1
	x3(0)=1
	x4(0)=0
	R2: $\{x1,x2\},\{x3\},\{x4\}$
	R2 is not a BBE (def 2.1)
	R2 is not a BBE (def 2.3)
	\end{comment}
	
	We first introduce the notion of \emph{constant} state on an equivalence relation $R$.
	
	\begin{definition}\textbf{(Constant State)}\label{defCS}
		A state $\textbf{s} \in \mathbb{B}^n$ is \emph{constant} on $R$ if and only if $\forall (x_i,x_j)\in R$ it holds that $s_{x_i}=s_{x_j}$.
	\end{definition}
	
	Consider our running example and an equivalence relation $R$ given by the partition $X_R = \{ \{ x_{Sp8},x_{Fgf8}\}$, $\mathit{\{x_{Pax6}\}, \{x_{Emx2}\}}$, $\mathit{\{x_{Coup\_tfi}\}\}}$.
	%$\mathit{R=}$ $\mathit{\{(Fgf8,Fgf8)}$, $\mathit{(Fgf8,Sp8)}$, $\mathit{(Sp8,Fgf8)}$, $\mathit{(Sp8,Sp8)}$, $\mathit{(Pax6,Pax6)}$, $ \mathit{(Emx2,Emx2)}$, $\mathit{(Couptfi,Couptfi)\}}$.
	%We remind that the first and the fourth positions in the BN states of the STG are evaluations of the variables $x_{Sp8},x_{Fgf8}$ respectively.
	The states constant on $R$ are colored in purple in Fig.~\ref{fig:image3}.
	For example, the state $\textbf{s}=(1,0,1,1,1)$ is constant on $R$  because $\mathit{s_{Sp8}=s_{Fgf8}}$ (the first and fourth positions of $\textbf{s}$, respectively). % which are the evaluations of the variables $\mathit{x_{Fgf8}, x_{Sp8}}$ respectively, and correspond to the 1st and 4th element of the vector $\textbf{s}$.
	On the contrary, $(1, 0, 1, 0, 1)$ is not constant on $R$. % because $\mathit{s_{Sp8} \neq s_{Fgf8}}$.
	%The relation $R$ of previous example induces the partition $\mathit{H=\{\{ Sp8,Fgf8\}}$, $\mathit{\{Pax6\}, \{Emx2\}}$, $\mathit{\{Couptfi\}\}}$.
	
	We now define \emph{Backward Boolean Equivalence (BBE)}.
	\begin{definition}\textbf{(Backward Boolean Equivalence)} \label{defBE}
		Let $B=(X,F)$ be a BN, $X_R$ a partition of the set $X$ of variables, and $C \in X_R$ a class of the partition.
		A partition $X_R$ is a \emph{Backward Boolean Equivalence (BBE)} if and only if the following formula is valid:
		\[
		%\tag{$\Phi^{X_R}$}
		\Phi^{X_R} \equiv
		\left(\operatorname*{\bigwedge}_{\stackrel{C \in X_R}{x,x' \in C}}  \Bigl( x = x' \Bigr) \right)
		\longrightarrow
		\operatorname*{\bigwedge}_{\stackrel{C \in X_R}{x,x' \in C}}  \Bigl( f_x(X) = f_{x'}(X
		) \Bigr)
		\]
	\end{definition}
	
	$\Phi^{X_R}$ says that if for all equivalence classes $C$ the variables in $C$ are equal, then the update functions of variables in the same equivalence class stay equal.
	
	In other words, $R$ is a BBE if and only if for all  $\textbf{s}\in \mathbb{B}^n$ constant on $R$ it holds that $F(\textbf{s})$ is constant on $R$. BBE is a relation where the update functions $F$ preserve the ``constant'' property of states. The partition $X_R = \{ \{ x_{Sp8},x_{Fgf8}\}$, $\mathit{\{x_{Pax6}\}, \{x_{Emx2}\}}$, $\mathit{\{x_{Coup\_tfi}\}\}}$ described above is indeed a BBE. This can be verified on the STG: all purple states (the constant ones) have outgoing transitions only towards purple states.
	
	%whenever the BN is in a state with $\mathit{s_{Sp8}=s_{Fgf8}}$ then all transitions lead to states where $\mathit{s_{Sp8}=s_{Fgf8}}$. Indeed, we see in Fig.~\ref{fig:image3} that purple boxes always lead to purple boxes.
	
	%It is worth to clarify that a BN may have several BBEs. This will be a matter of analysis in the next section where we explain how to check if a given partition is a BBE and how to compute the maximal BBE. We now focus on the way that BBEs can be used to reduce BNs.
	We now define the notion of BN reduced up to a BBE $R$. %BN reduction based on the BN quotient that is induced by the $BBE$ relation $R$.
	Each variable in the reduced BN represents one  equivalence class in $R$. We denote by $f\{^a/_b\}$ the term arising by replacing each occurrence of $b$ by $a$ in the %update 
	function $f$.
	
	\begin{definition}%\textbf{(BN reduction)}
		The reduction of $B$ up to $R$, denoted by
		$B/R$, is the BN $(X_R, F_R)$ where
		$F_R=\{f_{x_C} : C \in X_R\}$,
		with $f_{x_C}=f_{x_k}\{^{x_{C'}}/_{x_i}: \forall C'\in X_R, \forall x_i \in C'\}$
		for some $x_k\in C$.
	\end{definition}
	
	The definition above uses one variable per equivalence class, selects the update function of any variable in such class, and replaces all variables in it %such update function 
	with a representative one per equivalence class. Fig.~\ref{fig:reducedCAD} shows the reduction of the cortical area development BN. We selected the update function of %the variable
	$\mathit{x_{Sp8}}$ as the update function of the class-variable $\mathit{x_{\{{Fgf8},{Sp8}\}}}$, and replaced every occurrence of $\mathit{x_{Sp8}}$ and $\mathit{x_{Fgf8}}$ with %the variable
	$\mathit{x_{\{Fgf8,Sp8\}}}$.
	%It is worth to remark that as update function $f_{x_C}$ for an equivalence class $C$ we can select any update function $f_{x_C}$ from the members of $C$ since all variables that belong to the same equivalence class have equal updates on states that are constant on $R$.
	%
	The STG of such reduced BN is given in  Fig.~\ref{fig:image3}(b).
	
	\begin{figure}[h]
		\centering
		\vspace{-0.5cm}
		\scalebox{0.85}%[0.85]
		{%
			$
			\begin{array}{rcl}
				x_{\{\mathit{Fgf8,Sp8}\}}(t+1) &=& x_{\{\mathit{Fgf8,Sp8}\}}(t) \land \neg x_{\{Emx2\}}(t) \\
				x_{\{\mathit{Pax6}\}}(t+1) &=& \neg x_{\{\mathit{Emx2}\}}(t) \land x_{\{\mathit{Fgf8,Sp8}\}}(t) \land \neg x_{\{\mathit{Coup\_tfi}\}}(t)\\	
				x_{\{\mathit{Emx2}\}}(t+1) &=& \neg x_{\{\mathit{Fgf8,Sp8}\}}(t) \land \neg x_{\{\mathit{Pax6}\}}(t) \land \neg x_{\{\mathit{Fgf8,Sp8}\}}(t) \land x_{\{\mathit{Coup\_tfi}\}}(t) \\		
				x_{\{\mathit{Coup\_tfi}\}}(t+1) &=& \neg x_{\{\mathit{Fgf8,Sp8}\}}(t) \land \neg x_{\{\mathit{Fgf8,Sp8}\}}(t)\\
			\end{array}
			$
		}
		\caption{The BBE-reducion of %ed BN after the application of BN reduction to
			the cortical area development network of Fig.~\ref{fig:image2}.}
		\label{fig:reducedCAD}
	\end{figure}

	\subsection{Computation of the maximal BBE}\label{sec:compute}
	%\paragraph*{Compute the maximal BBE}
	%We have already presented a way to check if a given partition is a BBE.
	
	A crucial aspect of BBE is that it can be checked directly on a BN without requiring the generation of the STG. This is feasible by encoding the logical formula of Definition~\ref{defBE} into a logical SATisfiability problem~\cite{10.5555/1550723}. A SAT solver %the SMT solver Z3~\cite{de2008z3}. Z3 
	has the ability to check the validity of such a logical formula by checking for the unsatisfiability of its negation ($sat(\neg \Phi^{X_R})$). A partition $X_R$ is a BBE if and only if sat($\neg \Phi^{X_R}$) returns ``unsatifiable'', otherwise a counterexample (a witness) is returned, consisting of variables assignments that falsify $\Phi^{X_R}$. Using counterexamples, it is possible to develop a partition refinement algorithm that computes the largest BBE that refines an initial partition. %The overall procedure is implemented in ERODE~\cite{cardelli2017erode}.
	
	The partition refinement algorithm is shown in Algorithm~\ref{algorithm}. Its input are a BN and an initial partition of its variables $X$. A \emph{default} initial partition that leads to the maximal reduction consists of one block only, containing all variables. In general, the modeller may specify a different initial partition if some variables should not be merged together, placing them in different blocks. The output of the algorithm is the largest partition that is a BBE and refines the initial one.
	
	\begin{algorithm}[htp]
		\SetAlgoLined
		\KwResult{maximal BBE $H$ that refines $X_R$}
		%Import $F$ of a $BN=(X,F)$ and initial partition $X_R$\;
		H $\leftarrow$ $X_R$\;
		\While{true}{
			
			%\eIf{ $\Phi^{X_R}$ is valid}{
			\eIf{ $\Phi^{H}$ is valid}{
				return H \;
			}{
				$\mathbf{s} \leftarrow$ get a state that satisfy $\neg \Phi^{H}$\;
				$H^\prime \leftarrow \emptyset$\;
				\For{$ C \in H$}
				{$C_0 = \{x_i \in C : f_{x_i}(\mathbf{s}) = 0\}$\;
					$C_1 = \{x_i \in C : f_{x_i}(\mathbf{s}) = 1\}$\;
					$H' = H' \cup \{C_1\} \cup \{C_0\}$\; %\comax{singletons added}  \; 
					%$R \leftarrow \{(x_i,x_j):x_i,x_j \in H_l$ and $f_i(s)=f_j(s)\} $\;
					%$H^\prime \leftarrow H^\prime \cup (H_l/R)$\;
				}
				
				$H \leftarrow H' \setminus \{\emptyset\}$\; %\comax{setminus added}  \;
			}
		}
		\caption{Compute the maximal BBE that refines the initial partition $X_R$ for a BN $(X,F)$}
		\label{algorithm}
	\end{algorithm}
	
	We now explain how the algorithm works for  input the cortical area development BN and the %default 
	initial partition $\mathit{X_R=\{\{x_{Fgf8},x_{Emx2},}$ $\mathit{x_{Pax6},}$ $\mathit{x_{Sp8},x_{Coup\_tfi}\}\}}$.
	
	\paragraph{Iteration 1.}The algorithm enters the \emph{while} loop, and the solver %Z3 
	checks if %the formula 
	$\mathit{\Phi^{X_R}}$ is valid. % through the satisfiability of its negation. 
	$X_R$ is not a BBE, therefore %and consequently 
	the algorithm enters the second branch of the \emph{if} statement. The solver gives an example satisfying $\mathit{\neg \Phi^{X_R}}$:  $\mathit{s=(s_{x_{Fgf8}},s_{x_{Pax6}},s_{x_{Emx2}}}$, $\mathit{s_{x_{Sp8}},}$ $\mathit{s_{x_{Coup\_tfi}})=}(0,0,0,0,0)$. Since $\mathit{t}$ $\mathit{=F(s)}$  $=(0,0,$ $0,0,1)$, the \emph{for} loop partitions $\mathit{G}$ into $X_{R_1}=\mathit{\{\{x_{Fgf8}},$ $\mathit{x_{Pax6}}$, $\mathit{x_{Emx2}}$ $\mathit{x_{Sp8}\},\{x_{Coup\_tfi}\}\}}$. The state $t=(0,0,0,0,1)$ is now constant on $\mathit{X_{R_1}}$.
	
	\paragraph{Iteration 2.} The algorithm checks if $\mathit{\Phi^{X_{R_1}}}$ is valid (i.e. if $X_{R_1}$ is a BBE). $\mathit{X_{R_1}}$ is not a BBE. The algorithm gives a counterexample with   $s=(0,0,0,0,1)$ and $t=F(s)=(0,0,1,0,1)$. The \emph{for} loop refines $\mathit{X_{R_1}}$ into $X_{R_2}=\mathit{\{\{x_{Fgf8},x_{Pax6}}$ $\mathit{x_{Sp8}\},\{x_{Emx2}\},\{x_{Coup\_tfi}\}\}}$. $\mathit{X_{R_2}}$ %is a partition that 
	makes  $t=(0,0,1,0,1)$ constant.
	
	\paragraph{Iteration 3.}The algorithm checks if $\mathit{G_2}$ is a BBE. The formula $\neg \mathit{\Phi^{X_{R_2}}}$ is satisfiable,  so $\mathit{G_2}$ is not a BBE, and the solver %Z3 
	provides an example with  $s=(1,1,0,1,1)$ and $F(s)=(1,0,0,1,0)$. Hence, $\mathit{X_{R_2}}$ is partitioned into $X_{R_3}=\mathit{\{\{x_{Fgf8},}$ $\mathit{x_{Sp8}\},}$ $\mathit{\{x_{Pax6}\}}$ $\mathit{\{x_{Emx2}\},\{x_{Coup\_tfi}\}\}}$.
	
	\paragraph{Iteration 4.} The SAT solver proves that $\mathit{\Phi^{X_{R_3}}}$ is valid.
	
	The number of iterations needed to reach a BBE depends on the %random
	counterexamples that the SAT solver provides. %~\ref{algorithm}.
	As for all partition-refinement algorithms, it can be easily shown that the number of iterations is bound by the number of variables. Each iteration requires to solve a SAT problem which is known to be NP-complete, however we show in Section~\ref{sec4} that we can easily scale to the largest models present in popular BN repositories.

	%Similarly, if we initialize the algorithm with $G=\{\{x_{Fgf8},x_{Sp8},x_{Emx2}\}$, $\{x_{Pax6},x_{Couptfi}\}\}$, then the output partition will again be $\mathit{H=\{\{ x_{Sp8},x_{Fgf8}\},}$ $\mathit{ \{x_{Pax6}\}, \{x_{Emx2}\}}$, $\mathit{\{x_{Couptfi}\}\}}$. %The algorithm outputs the same partition $H$ for initial partition $G=\{\{x_{Fgf8},x_{Sp8}\}$, $\{x_{Pax6},x_{Couptfi},x_{Emx2}\}\}$. On the other hand, if we set $G=\{\{x_{Couptfi},$ $x_{Fgf8}\}$, $\{x_{Pax6},x_{Sp8},x_{Emx2}\}\}$, the algorithm outputs the partition $H=$ $\{ \{x_{Fgf8}\}$, $\{x_{Emx2}\}$,  $\{x_{Pax6}\}$, $\{x_{Sp8}\}$, $\{x_{Couptfi}\}\}$, meaning that no reduction is achieved.

	We first show that given an initial partition there exists exactly one \emph{largest} BBE that refines it.
	~\footnote{All proofs are given in Appendix~\ref{sec:proofs}}
	%~\footnote{All proofs are given in the extended version of this paper~\cite{bbeextended}.}

	After that, we prove that Algorithm~\ref{algorithm} indeed provides the maximal BBE that refines the initial one.
	
	%Finally, we show that given a $BN=(X,F)$ and an initial partition $G$, there exists a unique maximal BBE partition refining $H$.
	\begin{theorem}\label{th:maximal}
		Let $BN=(X,F)$ and $X_{R}$ a partition. There exists a unique maximal BBE $H$ that refines $X_{R}$.
	\end{theorem}
	
	\begin{theorem}\label{th:computesmaximal}
		Algorithm~\ref{algorithm} computes the maximal BBE partition refining $X_{R}$.
	\end{theorem}

	\subsection{Relating Dynamics of Original and Reduced BNs}\label{sec3.3}
	
	Given a BN $B$ and a BBE $R$, $STG(B/R)$ can be seen as the subgraph of $STG(B)$ composed of all states of $STG(B)$ that are constant on $R$ and their transitions. Of course, those states are transformed in $STG(B/R)$ by ``collapsing'' BBE-equivalent variables in the state representation.
	%In this section we formalize this to prove properties preserved by the reduction.
	%Before introducing the formalization we can see the relation
	This can be seen by comparing the STG of the our running example (left part of Fig.~\ref{fig:image3}) and of its reduction (right part of Fig.~\ref{fig:image3}). The states (and transitions) of the STG of the reduced BN correspond to the purple states of the original STG. % Fig.~\ref{fig:image3} displays the STG of the original cortical area development BN together with its constant states which are colored purple. The STG of the reduced cortical area development BN is depicted in~\ref{fig:image6}~. Every purple state in the STG of Fig.~\ref{fig:image3} correspond to a state in the STG of Fig.~\ref{fig:image6}~.
	
	\begin{comment}
	\begin{figure}[t]
	\centering
	\includegraphics[width=0.40\linewidth]{STGCADRed.pdf}
	\caption{The STG of the reduced cortical area development BN %of Fig. \ref{fig:image4}
	.}
	\label{fig:image6}
	\end{figure}
	
	\end{comment}
	
	Let $B$ be a BN with $n$ variables, $S \subseteq \mathbb{B}^n$ be the states of its STG, and $R$ a BBE for $B$. We use $S_{\mid R}$ to denote the subset of $S$ composed by all and only the states constant on $R$. With $STG(B)_{\mid R}$ we denote the subgraph of  $STG(B)$ containing $S_{\mid R}$ and its transitions. Formally  $STG(B)_{\mid R} = (S_{\mid R},T_{\mid R})$, where $T_{\mid R} = T \cap (S_{\mid R} \times S_{\mid R})$.
	
	The following lemma formalizes a fundamental property of $STG(B)_{\mid R}$, namely that all attractors of $B$ containing states constant on $R$ are preserved in $STG(B)_{\mid R}$.
	
	\begin{lemma}\textbf{(Constant attractors)}
		\label{lemma:constantAttractors}
		Let $B(X,F)$ be a BN, $R$ be a BBE, and $A$ an attractor. If $A \cap S_{\mid R} \neq \emptyset$ then $A \subseteq S_{\mid R}$ .
	\end{lemma}
	
	%\coma{What about the other direction? Can there be attractors in the reduced models that 'are not attractors in the original one?}\coma{We believe we don't add attractors. We should state this in a lemma/theorem.}
	
	We now define the bijective mapping $m_R : S_{\mid R} \leftrightarrow S_R$ induced by a BBE $R$, where $S_R$ are the states of $STG(B/R)$, as follows:
	$m_R(\textbf{s}) = (v_{C_1}, \dots , v_{C_{|X/R|}})$ where $v_{C_j} = s_{x_i}$ for some $x_i \in C_j$.
	In words $m_R$ bijectively maps each state of $STG(B)_{\mid R}$ to their compact representation in $STG(B/R)$.
	Indeed, $STG(B)_{\mid R}$ and $STG(B/R)$ are isomorphic, with $m_R$ defining their (bijective) relation.  We can show this through the following lemma.
	
	\begin{lemma}\textbf{(Reduction isomorphism)}
		\label{lemma:isomorphism}
		Let $B(X,F)$ be a BN and $R$ be a BBE. Then, it holds
		\begin{enumerate}
			\item For all states $\textbf{s} \in S_{\mid R}$ it holds
			$
			F_R(m_R(\textbf{s})) = m_R(F(\textbf{s})).
			$
			\item For all states $\textbf{s} \in S_{R}$ it holds
			$
			F(m^{-1}_R(\textbf{s})) = m^{-1}_R(F_R(\textbf{s})).
			$
		\end{enumerate}
	\end{lemma}
	
	The previous Lemma ensures that BBE does not generates spurious trajectories or
	attractors in the reduced system. We can now state the main result of our approach, namely that the BBE reduction of a BN for a BBE $R$ exactly preserves all attractors that are constant on $R$ up to renaming with $m_R$.
	
	\begin{theorem}\textbf{(Constant attractor preservation)}\label{th:attractorPreservation}
		Let $B(X,F)$ be a BN, $R$ a BBE, and $A$ an attractor. If $A \cap S_{\mid R} \neq \emptyset$ then $m_R(A)$ is an attractor for B/R.
	\end{theorem}
	
	%Essentially, the theorem secures that BBE preserves the length of attractors. In other words, a constant attractor has the same length in the reduced BN as in the original BN.
	
	%\input{olddynamicspreservation}
	
	%\begin{figure}
	%	\subfloat[the original STG]{
	%	\includegraphics[scale=0.3]{STGCorticalC.png}
	%	}	
	%	\subfloat[the STG of the reduced BN]{
	%	\includegraphics[scale=0.50]{STGCorticalRed1.png}
	%	}
	%\caption{(a): the STG of the original cortical area development BN, (b): the STG of the reduced cotical area development BN}
	%\label{fig:image6}
	%\end{figure}
	
	%\begin{lemma}
	%	\comg{Explain what happens to attractors, transient trajectories etc.}
	%\end{lemma}
	
	%\commentAlberto{A ``user'' of BNs that reads the current content of the section could wonder what are the consequences for the analysis, so yes, I would tend to think that we should say what happens to attractors and other properties. I would first identify the properties of interest, how/if they are preserved. I suggested to do so a long time ago. Some properties are mentioned in the preliminaries (but not explained). I would explain them (in the preliminaries) and then explain here what happens to them.}

	\section{Application to BNs from the Literature}\label{sec4}
	We hereby apply BBE to BNs from the GINsim repository. %~(\url{http://ginsim.org/models\_repository}).
	%We begin by describing our prototype toolchain (Section~\ref{sec:toolchain}) that makes available our results to the BN community.
	Section~\ref{sec:42} validates BBE on all models from the repository, while Section~\ref{casestudies} studies the runtime speedups brought by BBE on attractor-based analysis of selected case studies, showing cases for which BBE makes the analysis feasible. 
	%%
	%\footnote{These models are further analysed in Appendix~C of~\cite{bbeextended}  using  initial partitions based on information from the original publications, obtaining better reductions. % leading to reductions intermediate to the maximal and ID one.
	%}
	%%
	\footnote{These models are further analysed in Appendix~\ref{sec:appendix_casestudies} using  initial partitions based on information from the original publications, obtaining better reductions. % leading to reductions intermediate to the maximal and ID one.
	}
	%
	%We focus on attractor identification to selected case studies with respect to the results of \ref{sec3.3}~, and assess the computational efficiency.
	Section~\ref{comparisonpnas} compares BBE with the approach based on ODE encoding from~\cite{cardelli2017maximal}, showing how  such encoding leads to scalability issues and to the loss of reduction power.
	\footnote{
		Appendix~\ref{app:speedups} 
		%Appendix~D of~\cite{bbeextended}
		further studies BBE-induced runtime speedups to STG generation on the repository. We display again cases where BBE makes the analysis feasible.  % despite being intractable for the original model.
	}
	
	\begin{figure}[t]
		\centering
		\includegraphics[width=0.95\linewidth]{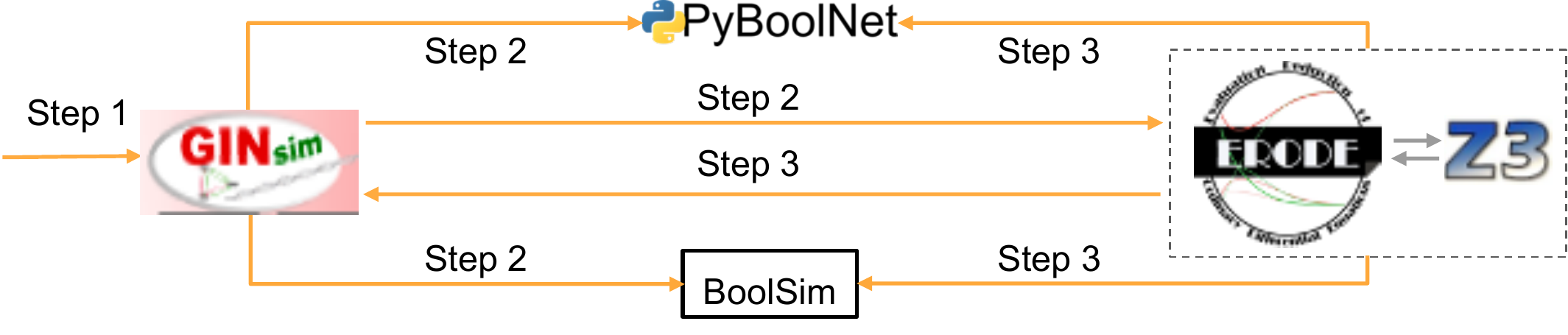}
		\caption{BBE toolchain.
			(Step 1) We  use GINsim~\cite{chaouiya2012logical} to access its model repository, and (Step 2) export it in the formats of the other tools in the toolchain to perform: STG generation (PyBoolNet~\cite{pyboolnet}), attractor analysis (BoolSim~\cite{di2007dynamic}), and BBE reduction (ERODE~\cite{cardelli2017erode}).
			%		(Step 3) We export BNs in an input format of ERODE, which essentially coincides with the one of PyBoolNet, enabling BBE-reduction using an internal integration with Z3.
			(Step 3) We export the reduced models for analysis to PyBoolNet and BoolSim, or to GINsim.
		}
		\label{fig:toolchain}
	\end{figure}
	
	%\subsection{ToolChain}\label{sec:toolchain}
	%In order to reduce real-world models,
	The experiments have been made possible by a novel  toolchain (Fig.~\ref{fig:toolchain}) combining tools from the COLOMOTO initiative~\cite{naldi2015cooperative}, and the
	reducer tool ERODE~\cite{cardelli2017erode} which was extended here to support BBE-reduction. For Algorithm~\ref{algorithm} we use the solver Z3~\cite{de2008z3} which was already integrated in ERODE.

	All experiments %in this Section
	were conducted on a common laptop with an Intel Xeon(R) %CPU E3-1505M
	2.80GHz and $32$GB of RAM.
	We imposed an arbitrary %Due to high complexity of the processes involved, we imposed a
	timeout of 24 hours for each task, after which we terminated the analysis.
	We refer to these cases as \emph{time-out}, while we use \emph{out-of-memory} if a tool terminated with a memory error. %All runtime are given in seconds.
	
	\subsection{Large Scale Validation of BBE on BNs}\label{sec:42}
	
	%\paragraph{Hypothesis}
	We validate BBE on real-world BNs in terms of the number of BNs that can be reduced and the average reduction ratio. %We show  that BBE reduces a significant number of cases and we would consider the method successful if the reduced BN has $4/5$ of the size of the original.
	
	\paragraph{Configuration.}
	We conducted our investigation on the whole GINsim model repository which contains 85 networks: 29 are Boolean, and 56 are multivalued.
	In multivalued networks (MNs), some variables have more than 2 activation statuses, e.g. $\{0,1,2\}$. These models are automatically \emph{booleanized}~\cite{delaplace2020bisimilar,chaouiya2013sbml} by GinSim when exporting in the input formats of the other tools in the tool-chain.
	% \comal{I fear this may be too short. Also you don't transform only the variables but also the functions. I would say something like ``For multi-valued BNs we applied the booleanization supported by Ginsim  (see [cite delaplace2020bisimilar] ), which transforms a multi-valued BN into an equivalent boolean BN}. GinSim implements this transformation.
	%This is implemented by biolqm software \cite{bibid} which is incorporated in GINsim.
	%For a nice review about the booleanization of multivalued networks one can read \cite{delaplace2020bisimilar}.
	
	Most of the models in the repository have a specific structure \cite{naldi2012efficient} where a few variables are so-called \emph{input variables}. These are variables whose update functions are either a stable function (e.g. $x(t+1)=0$, $x(t+1)=1$) or the identity function (e.g. $x(t+1)=x(t)$). These are named `input' because their values are explicitly set by the modeler to perform experiments campaigns.
	%Fig. \ref{merged1} displays the TCR-TLR merged BN of \cite{rodriguez2019cooperation} that we consider again later in Section \ref{casestudies}~ as a case study. Grey variables have update function $x(t+1)=1$, and pink variables have update function $x(t+1)=0$. The brown variables  $\mathit{TLR5,CD28,TCR}$ on top of the Fig. have identity update function, i.e. $x(t+1)=x(t)$.
	%In terms of the STG, there are no transitions between states that differ in values of input variables.
	We investigate two reduction scenarios relevant to  input variables. In the first one, Algorithm~\ref{algorithm} starts with initial partitions that lead to the \emph{maximal reduction}, i.e. consisting of one block only. In the second scenario, we provide initial partitions that isolate inputs in singleton blocks. Therefore, we prevent their aggregation with other variables, and obtain reductions independent of the values of the input variables (we recall that BBE requires related variables to be initialized with same activation value). We call this case \emph{input-distinguished (ID) reduction}.
	%In the latter, some trivial cases: either only input variables appear to be backward equivalent, or the target variables of inputs appeared to be backward equivalent together with the input variables. In some cases, the latter results in reductions with large equivalence classes consisting of input variables and their descentants. This phenomenon is apparent in Fig. \ref{merged1}~: The input variable $\mathit{TLR5}$ forces its descentants with its value, so we obtained the chain of brown background variables.
	
	\paragraph{Results.}
	By using the maximal reduction setting, we obtained reductions on 61 of the 85 models, while we obtained ID reductions on 38 models.
	We summarize the reductions obtained for the two settings in Fig.~\ref{charts}, displaying the distribution of the reduction ratios $r_m=N_m/N$ and $r_i=N_i/N$, where $N$, $N_m$ and $N_i$ are the number of variables in the original BN, in the maximal BBE-reduction, and in the ID one, respectively.
	\footnote{More details can be found in 
		%Table~2 of Appendix~B of~\cite{bbeextended}.
		Table~\ref{table with results} of Appendix~\ref{largescaletable}.
	}
	% 
	%The charts also provide the average reduction ratios obtained for each setting, further detailed for BN and MN in the table in the bottom-right.
	We also provide the average reduction ratios on the models, showing that it does not substantially change across Boolean or multivalued models.
	No reduction took more than 3 seconds.

	%\begin{figure}%[H]
	%%	\centering
	%%	\begin{tabular}{cc}
	%		\hspace{-0.5cm}
	%		\includegraphics[width=1.05\linewidth]{MAXchart.png}%[width=12cm,height=6cm]
	%		\\
	%\vspace{-0.5cm}		
	%\begin{center}
	%%		\hline
	%%		\\
	%		\includegraphics[width=0.57\linewidth]{IDchart.png}
	%		\quad
	%			\scalebox{0.7}{
	%							
	%				\begin{tabular}{ccc}
	%					\multicolumn{3}{c}{Average ratios of BNs and MNs}\\
	%					& \emph{Maximal} & \emph{ID}  \\
	%					\hline
	%					BNs  &  0.66& 0.83  \\
	%					\hline
	%					MNs  & 0.68 & 0.95
	%					\vspace{5cm}
	%				\end{tabular}
	%			}
	%%			\caption{The average reduction ratio of all BNs and MNs}
	%%			\label{ratioMBN}
	%\end{center}
	%%	\end{tabular}
	%\vspace{-2.5cm}
	%	\caption{
	%		Reduction ratios (reduced variables over original variables) on the models from the GINsim repository using the maximal (top) and input-distinguished (bottom) reduction strategy.
	%		Each bar refers to a BN: the x-axis contains the model identifier (MI) from Table~\ref{table with results}. In each chart, we provide the average reduction ratios in terms of an horizontal line, further detailed for BN and MN in the table in the bottom-right.}
	%	\label{charts}
	%\end{figure}
	
	%\begin{figure}[t]
	%\centering
	%\includegraphics[width=0.4\linewidth]{experiments/Hist_LargeScale.pdf}	
	%		\quad
	%\scalebox{0.7}{
	%	\begin{tabular}{ccc}
	%		\multicolumn{3}{c}{Average ratios of BNs and MNs}\\
	%		& \emph{Maximal} & \emph{ID}  \\
	%		\hline
	%		BNs  &  0.66& 0.83  \\
	%		\hline
	%		MNs  & 0.68 & 0.95\\
	%		ALL & 0.67 & 0.91
	%%		\vspace{5cm}
	%	\end{tabular}
	%}
	%\end{figure}

	\begin{figure}
		\centering
		%\subfloat[Figura 1]{
		\includegraphics[width=0.5\linewidth]{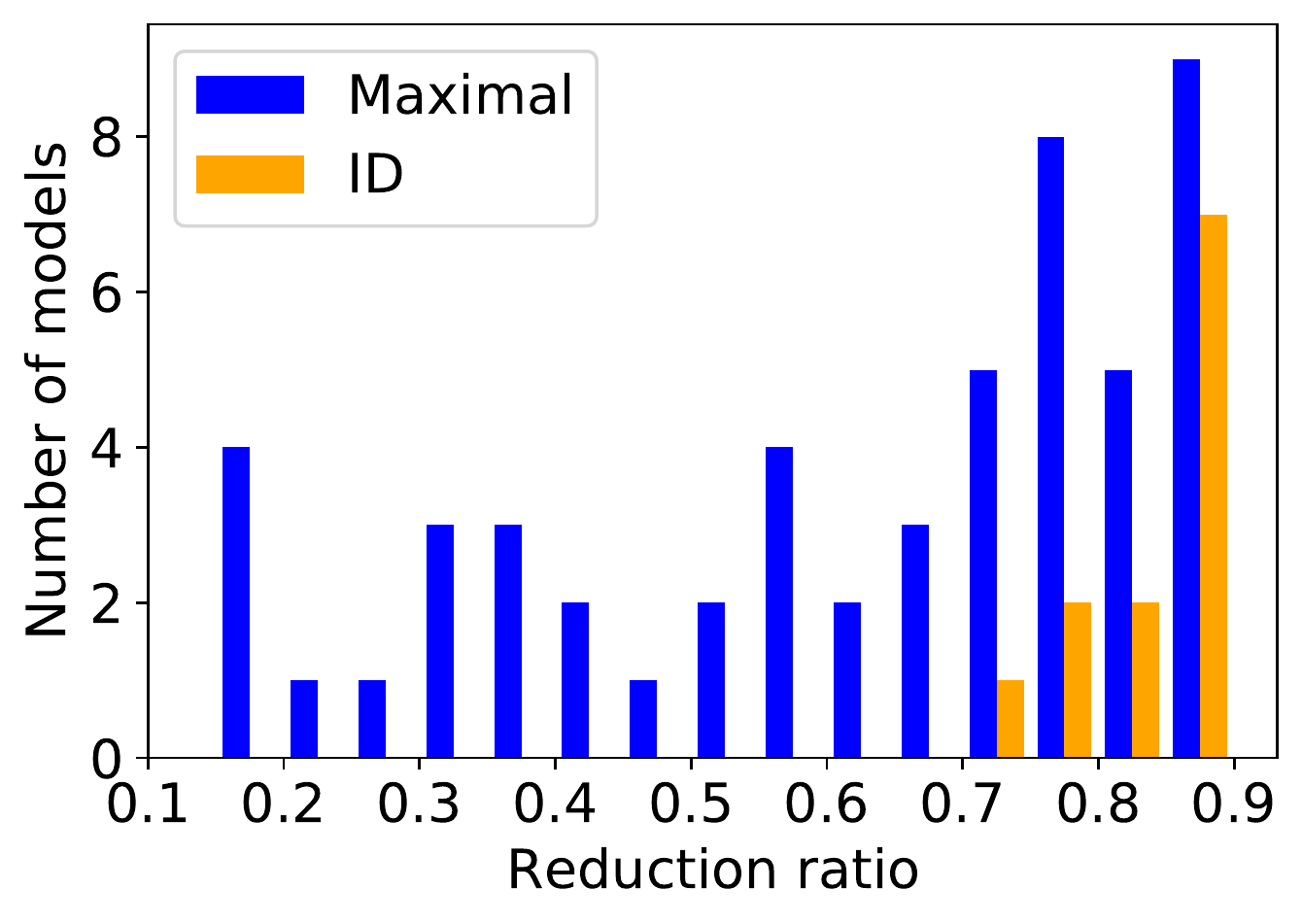}
		%}
		\qquad
		%\subfloat[Figura 2]{	
		%	\quad
		\scalebox{1.0}{
			\begin{tabular}{ccc}
				\multicolumn{3}{c}{Average reduction ratios}\\
				\toprule
				& \emph{Maximal} & \emph{ID}  \\
				%		\hline
				BNs  &  0.66& 0.83  \\
				%			\hline
				MNs  & 0.68 & 0.95\\
				ALL & 0.67 & 0.91\\
				\bottomrule
				\vspace{4.2cm}
			\end{tabular}
		}
		\vspace{-3cm}
		\caption{(Left) Distribution of reduction ratios (reduced variables over original ones) on all models from the GINsim repository using the maximal and ID reduction strategy.
			Each bar counts the number of models with that reduction ratio, starting from 15\% up to 90\%, with step 5\%.
			(Right) Average reduction ratios for Boolean, Multivalued and all models.\label{charts}}
		%}
	\end{figure}
	
	%\begin{minipage}{0.4\textwidth}
	%\begin{figure}
	%%	\centering
	%	\includegraphics[width=0.4\linewidth]{experiments/Hist_LargeScale.pdf}	
	%	\quad
	%	\scalebox{0.7}{
	%		\begin{tabular}{ccc}
	%			\multicolumn{3}{c}{Average ratios of BNs and MNs}\\
	%			& \emph{Maximal} & \emph{ID}  \\
	%			\hline
	%			BNs  &  0.66& 0.83  \\
	%			\hline
	%			MNs  & 0.68 & 0.95\\
	%			ALL & 0.67 & 0.91
	%			%		\vspace{5cm}
	%		\end{tabular}
	%	}
	%\end{figure}
	%\end{minipage}

	%\begin{table}[h]
	%	\centering
	%	\scalebox{0.75}{
	%		\begin{tabular}{c|c|c}
	%			
	%			 & \emph{Maximal} & \emph{Input-distinguished}  \\
	%			\hline
	%			BNs  &  0.66& 0.834  \\
	%			\hline
	%			MNs  & 0.682 & 0.956
	%		\end{tabular}
	%	}
	%	\caption{The average ratio of all BNs and MNs}
	%\label{ratioMBN}
	%\end{table}
	
	\paragraph{Interpretation.}
	BBE reduced a large number of models (about $72\%$). In particular, this happened in $24$ out of the $29$ ($83\%$) Boolean models and in $37$ out of $56$ ($66\%$) multivalued networks. The average reduction ratio for the maximal and ID strategies are $0.67$ and $0.91$, respectively. % \coma{$0.85$ in the CSV you have sent to me it says 0.91}, respectively.
	For the former strategy, we get trivial reductions in 22 models wherein only input variables are related. In such trivial cases, the ID strategy does not lead to reduction. %with the Boolean models having slightly better reduction ratios on average. %. In these cases it holds that $\mathit{N=N_i}$ and $\mathit{N >N_m}$.
	In other cases, the target variables of inputs (i.e. variables with incoming edges only from input variables considering the graphical representation of variables) appeared to be backward equivalent together with the input variables. This results in reductions with large equivalence classes consisting of input variables and their descendants. %We see an example of this in Section~\ref{casestudies} (see the chain of brown variables in the top-left of Fig.~\ref{merged}).
	These are interesting reductions which get lost using the ID approach, as the input variables get isolated.

	\subsection{Attractor analysis of selected case studies}\label{casestudies}
	
	\paragraph{Hypothesis.} We now investigate the fate of asymptotic dynamics after BBE-reduction, and test the computational efficiency in terms of time needed for attractor identification in the original and reduced models. %, and \comg{this sentence confuses me} exploit use reasoned input partitions to obtain  reduced BNs smaller than the ID case while preserving analysis of interest in the original paper.
	We expect that BBE-reduction can be utilized to (i) gain fruitful insights into large BN models and (ii) to reduce the time needed for attractor identification.
	%We achieve the former (i) by exploiting the results of the ID and maximal reduction to obtain refined reduced models and the latter (ii) by reducing the number of biologically relevant attractors.
	
	\paragraph{Configuration.}
	Our analysis focuses on three BNs from the GINsim repository. The first is the Mitogen-Activated Protein Kinases (MAPK) network~\cite{grieco2013integrative} with 53 variables. The second refers to the survival signaling in large granular lymphocyte leukemia (T-LGL)~\cite{zhang2008network} and contains 60 variables. The third is the merged Boolean model~\cite{rodriguez2019cooperation} of T-cell and Toll-like receptors (TCR-TLR5) which is the largest BN model in GINsim repository with 128 variables.

	\paragraph{Results.}
	The results of our analysis are summarized in Table~\ref{rt} for the original, ID- and maximal-reduced BN. We present the number of variables (\emph{size}) and of Attractors (\emph{Attr.}), the time for attractor identification on the original model (\emph{An. (s)}) and that for reduction plus attractor identification (\emph{Red. + An.~(s)}). % The \emph{Analysis(s)} refers to the computation time of attractors.
	
	\begin{table}[h]
		\centering
		\scalebox{0.95}{
			\begin{tabular}{c rrr  rrr  rrr }
				%\cmidrule(r){2-10}
				\toprule
				\multicolumn{4}{c}{\emph{Original model}}     &  \multicolumn{3}{c}{\emph{ID reduction}}&  \multicolumn{3}{c}{\emph{Maximal reduction}} \\
				\cmidrule(r){1-4} \cmidrule(r){5-7} \cmidrule(r){8-10}
				& \emph{Size}& \emph{Attr.} & \emph{An.(s)} &\emph{Size}& \emph{Attr.} & \multicolumn{1}{c}{\emph{Red.+An.(s)}}&\emph{Size}& \emph{Attr.} &\multicolumn{1}{c}{ \emph{Red.+An.(s)}}\\
				\cmidrule(r){2-4} \cmidrule(r){5-7} \cmidrule(r){8-10}
				%\hline
				%TCRsig&40& 8 & 0,148s & 0,673s&31  &8 & 0,109s\\
				%\hline
				MAPK Network	& 53 &40 										& 16.50      &46  &40  & 15.33 &39 &17 & 3.49
				\\
				%\hline
				T-LGL 			&60  &264 										& 123.43     &57  &264 & 86.84  &52 &6 &3.49
				\\
				%\hline
				TCR-TLR  &128 &\multicolumn{2}{c}{---\emph{Time Out}---}   & 116 &\multicolumn{2}{c}{---\emph{Time Out}---}   & 95 &2 & 31.29 \\
				\bottomrule
			\end{tabular}
		}
		\caption{Reduction and attractor analysis on 3 selected case studies.}
		\label{rt}
	\end{table}
	
	\paragraph{Interpretation.}
	%We know that 
	ID reduction preserves all attractors reachable from any combination of activation values for inputs. % variables. 
	This is an immediate consequence of~\ref{lemma:isomorphism}, Theorem~\ref{th:attractorPreservation} and the fact that number of attractors in the original and the ID reduced BN is the same (see Table~\ref{rt}). Maximal reduction might discard some attractors.
	%We see that the 3 models can be reduced by both strategies.
	We also note that, despite the limited reduction in terms of obtained number of variables, we have important analysis speed-ups, up to two orders of magnitude. Furthermore, the largest model could not be analyzed, while it took just 30 seconds to analyze its maximal reduction identifying 2 attractors.
	\footnote{There might be further attractors of interest in addition to these. In 
		Appendix~\ref{sec:appendix_casestudies} 
	%Appendix~C of~\cite{bbeextended}
	we show how BBE could be used by a modeler by imposing ad-hoc initial partitions to preserve more attractors while reducing more than with the ID stategy.}
	%
	
	% The computational costs is much better in the last two cases. Notably, we failed to compute the attractors of the original TCR-TLR merged model with $\mathit{128}$ variables, but we managed to compute them in the case of maximal reduction. Indeed, the maximal BBE-reduced BN contains $\mathit{95}$ variables and Boolsim identifies $\mathit{2}$ attractors in $\mathit{31,294}$ seconds.
	
	%\se
	
	\subsection{Comparison with ODE-based approach from~\cite{cardelli2017maximal} }\label{comparisonpnas}
	As discussed, BBE is based on the backward equivalence notion firstly provided for ordinary differential equations (ODEs), chemical reaction networks, and Markov chains~\cite{DBLP:conf/popl/CardelliTTV16,cardelli2017maximal}.
	Notably,~\cite{cardelli2017maximal} shows how the notion for ODEs can be applied indirectly to BNs via an \emph{odification} technique~\cite{odification} to encode BNs as ODEs. Such odification transforms each BN variable into an ODE variable that takes values in the continuous interval [0,1]. The obtained ODEs preserve the attractors of the original BN because the equations of the two models coincide when all variables have value either 0 or 1. However, infinitely more states are added for the cases in which the variables do not have integer value.
	
	\begin{wrapfigure}{r}{0.35\textwidth}
		%\vspace{-1.15cm}
		\vspace{-0.4cm}
		\begin{center}
			\includegraphics[width=0.35\textwidth]{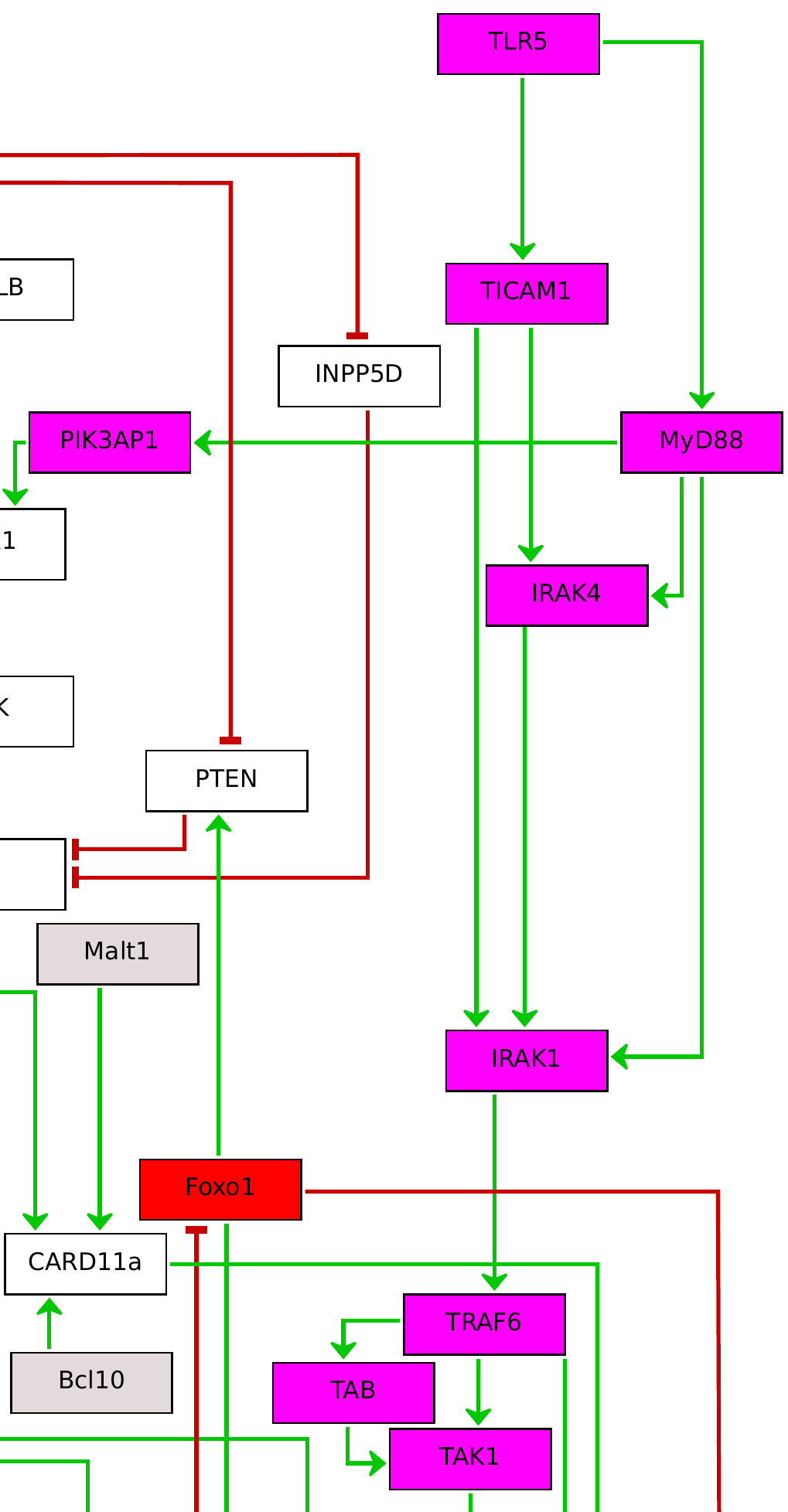}
		\end{center}
		\vspace{-0.4cm}
		\caption{Excerpt of GINsim's depict of TCR-TLR.\label{fig:tcrtlrexcerpt}}
		\vspace{-1.2cm}
	\end{wrapfigure}
	\paragraph{Scalability.}
	The technique from~\cite{cardelli2017maximal} has been proved able to handle models with millions of variables. Instead, the odification technique is particularly computationally intensive. Due to this, it failed on some models from the GINsim repository, including two from~\cite{20}, namely \emph{core\_engine\_budding\_yeast\_CC} and \emph{coupled\_budding\_yeast\_CC}, consisting of 39 and 50 variables, respectively. Instead, BBE could be applied in less than a second.
	
	\paragraph{Reduction power.}
	Another example is the \emph{TCR-TLR} model from the previous section. In this case, both the ODE-based and BBE techniques succeeded. However, BBE led to better reductions due to the added non-integer states in the ODEs. Intuitively, the ODE-based technique \emph{counts} incoming influences from equivalence classes of nodes, while BBE only checks whether at least one of such influence is present or not.
	%fact that the two models coincide only when ODE-based techniques 'counts' incoming influences equivalence classes of nodes, while BBE only checks whether there are influences or not.
	Figure~\ref{fig:tcrtlrexcerpt} shows an excerpt of the graphical representation of the model by GINsim. We use background colors of nodes to denote BBE equivalence classes (white denotes singleton classes). We see a large equivalence class of magenta species, 3 of which (\emph{IRAK4}, \emph{IRAK1}, and \emph{TAK1}) receive two influences by magenta species, while the others receive only one.
	This differentiates the species in the ODE-based technique,  keeping only the top four in the \emph{magenta} block, while all the others end up in singleton blocks. %Intuitively, the ODE encoding adds a term to the equation of a variable for each incoming contribution, while
	We compare the original %and ODE
	equations of \emph{MyD88} and  \emph{IRAK4} which have 1 and 2 incoming influences each.
	\begin{align*}
		x_\mathit{MyD88}(t+1) &= x_\mathit{TLR5}(t) %\label{eq:myd88}
		\\
		x_\mathit{IRAK4}(t+1) &= (\neg x_\mathit{MyD88}(t) \wedge x_\mathit{TICAM1}(t)) \vee (x_\mathit{MyD88}(t)) 	%\label{eq:irak4}
		%	\\
	\end{align*}
	We see that the two variables are BBE because their update functions depend only on the BBE-equivalent variables \emph{TLR5} and \emph{MyD88}, respectively. For \emph{IRAK4}, the three variables in the update function are BBE. Therefore, they have same value allowing us to simplify the update function to just \emph{MyD88}.
	The ODEs obtained for the 2 variables are, where $x_{-}'$ denotes the derivative of $x_{-}$:
	\begin{align*}
		x_\mathit{MyD88}' &= x_\mathit{TLR5} - x_\mathit{MyD88}\\
		x_\mathit{IRAK4}' &= x_\mathit{MyD88} + x_\mathit{TICAM1} -x_\mathit{MyD88}\cdot x_\mathit{TICAM1}   -x_\mathit{IRAK4}
	\end{align*}
	Given that all variables appearing in the equations are backward equivalent,  the two equations coincide with the original ones when all variables have values either 0 or 1. However, they differ for non-integer values.
	For example, in case all variables have value $0.5$, we get $0$ for the former, and $0.25$ for the latter.
	
	% \begin{table}[H]
	% 	\centering
	% 	\scalebox{1.0}{
	% 		\begin{tabular}{cr rr rr}
	% 			%\cmidrule(r){2-10}
	% 			\toprule
	% 			\multicolumn{2}{c}{\emph{Original model}}     &  \multicolumn{2}{c}{\emph{BBE}}&
	% 			\multicolumn{2}{c}{\emph{ODE-based}} \\
	% 			\cmidrule(r){1-2} \cmidrule(r){3-4} \cmidrule(r){5-6}
	% 			\emph{Model}& \emph{Size}& \emph{ID} &\emph{Max} &\emph{ID}& \emph{Max}   \\
	% 			\cmidrule(r){1-2} \cmidrule(r){3-4} \cmidrule(r){5-6}
	% %core_engine_budding_yeast_CC.ginml
	% Core budding yeast cell cycle & 39
	% &37&31&\multicolumn{2}{c}{\emph{Out of Memory}}
	% \\
	% %coupled_budding_yeast_CC.ginml			
	% budding yeast core+checkpoint & 50
	% &49&41&\multicolumn{2}{c}{\emph{Out of Memory}}
	% \\
	% 			%\hline
	% TCR-TLR merged  & 128 & 116 & 86 & %95 &
	% 111 & 94
	% 			\\
	% 			\bottomrule
	
	% 		\end{tabular}
	% 	}
	% 	\caption{Comparison with ODE-based reduction from~\cite{cardelli2017maximal}.}
	% 	\label{table:comparepnas}
	% \end{table}

	\section{Related Work}\label{sec:related}
	
	BN reduction techniques belong to three families according to their domain of reduction: (i) they reduce at syntactic level (i.e. the BN \cite{naldi2011dynamically,veliz2011reduction,richardson2005simplifying,bilke2001stability,naldi2012efficient,saadatpour2013reduction,Zanudoreduction}), (ii) at semantic level (i.e. the STG \cite{figueiredo2016relating,berenguier2013dynamical}),
	% These two consider the STG as labbeled transition system and perform reduction using bisimulation \cite{bibid}.
	or (iii) they transform BNs to other formalisms like Petri Nets \cite{chaouiya2008petri,steggles2007qualitatively} and ordinary differential equations \cite{wittmann2009transforming} offering formalism-specific reductions. However, (semantic) STG-reduction does not solve the state space explosion whereas the transformation to other formalisms has several drawbacks as shown in Section~\ref{comparisonpnas}.
	
	Syntactic level reduction methods usually perform variable absorption~\cite{bilke2001stability,naldi2011dynamically,veliz2011reduction,saadatpour2013reduction} at the BN. %Some semantic reduction methods consider STGs as Labelled Transition Systems (LTSs) \cite{figueiredo2016relating} and achieve reduction with bisimulation.
	%Specifically, in \cite{krause2010action} update timescales are incorporated in the transitions of LTS before performing bisimulation.
	%Another semantic method achieves STG compression by defining Hierarchical Transition Graph (HTG) \cite{berenguier2013dynamical}. The authors of \cite{berenguier2013dynamical} developed an algorithm that compacts the STG and automatically generates a compressed representation which emphasizes on relevant transient and asymptotic dynamics. On the contrary, BBE performs syntactic level reduction.
	%
	BN variables can get absorbed by the update functions of their target variables by replacing all occurrences  of  the  absorbed  variables  with  their  update  functions.
	This method was first investigated in~\cite{naldi2011dynamically} wherein update functions are represented as ordinary multivalued decision diagrams.
	The authors consider multivalued networks with updates being applied asynchronously and %However, their approach possess two restrictions: at each time step only one variable changes its value, and the value of the updated variable changes only by 1 (increases or decreases). %i.e. a randomly selected node is updated at each time transition.
	iteratively implement absorption. The process, despite preserving steady states in all synchronization schemas~\cite{veliz2011reduction}, might lead to loss of cycle attractors in the synchronous schema. % and/or degenerate their lengths.
	However, absorption of variables might lead to introduction of new attractors in the asynchronous case, i.e., by reducing the number of variables the number of attractors can stay the same or increase (attractors can split or new attractors can appear).
	%The relation between the asymptotic dynamics of the original model and those of the reduced models is formally established. On the other hand, we formally formally establish relations between the whole STGs of the original and the reduced model.
	%Absorption of regulatory components may cause spurious behaviours in the asynchronous case that considered in \cite{naldi2011dynamically}: complex attractors, that do not exist in the original model, may appear in the reduced one.  Instead of causing spurious behaviours, we only miss some of them as shown in Theorem~\ref{th:attractorPreservation}. The missed behaviors are those wherein BN states are not constant to BBE equivalent-variables.
	%On the contrary, we formally establish connections of the original and the reduced models, preserving all attractors constant on the used BBE. % that refer to the whole STG.
	
	A similar study~\cite{veliz2011reduction} presents a reduction procedure and proves that it preserves steady states. This procedure includes two steps. The first refers to the deletion of links between variables on their network structure. Deletion of pseudo-influences is feasible by simplifying the Boolean expressions in update functions. %This step is useful and flexible because (i) it may identify variables with stable values, (ii) it can be incorporated in any reduction procedure, (iii) can be combined with other reduction methods, and (iv) it does not affect the model dynamics. 
	The second step of the procedure refers to the absorption of variables like in~\cite{naldi2011dynamically}.
	
	The difference between studies \cite{veliz2011reduction}, \cite{naldi2011dynamically} is that~\cite{veliz2011reduction} exploits Boolean algebra instead of multivalued decision diagrams to explain absorption. Moreover, they refer only to Boolean networks, and do not consider any update schema.
	%Consequently, the results of~\cite{veliz2011reduction} could be used regardless the update schema, the number of variables updated, or the magnitude of change in the variables' values. Despite the fact that Z3 and Boolsim incorporate high level Boolean algebra techniques and binary decision diagrams, our study uses neither of them at first hand. %However, these techniques can be combined and incorporated in BBE-reduction.
	In studies~\cite{naldi2011dynamically,veliz2011reduction,saadatpour2013reduction}, self-regulated BN variables (i.e. variables with a self-loop in the graphical representation) can not be selected for absorption. The inability to absorb self-regulated variables is inherent in the implementation of absorption in contrast to our method where the restrictions are encoded by the user at the initial partition and self-regulated variables can be merged with other variables.
	
	In \cite{saadatpour2013reduction} the authors presented a two step reduction algorithm. The first step includes the absorption of input variables with stable function and the second step the absorption of single mediator variables (variables with one incoming and outgoing edge in the signed interaction graph). 
	%For asynchronous BNs, fixed points are preserved, while complex attractors are not in general. %\footnote{See author's note on \emph{mediator\_reduction} in~\cite{mediator}. %\url{https://github.com/jcrozum/StableMotifs/blob/master/Manual.md}}
	%They provide a rigorous mathematical proof that the algorithm conserves steady states and complex attractors of asynchronous BNs. 
	%
	%With respect to the discussion of  \cite{veliz2011reduction}, \cite{naldi2011dynamically}, variable absorption preserves reachability properties but degenerates the STG structure of the original BN. Our reduction method preserves reachability by drawing explicit connections of how the STGs of the original and the reduced BN are related. We do not degenerate the STG of the original BN since the reduced BN is a subgraph of the STG of the original BN as proves the isomorphism Lemma~\ref{lemma:isomorphism}. % In other words, the STG of the reduced BN contains all states that are constant on the equivalence relation induced by the BBE-variables and the transitions between these states.
	The first step of the algorithm in \cite{saadatpour2013reduction} is equally useful and compatible with the first step of \cite{veliz2011reduction}. Moreover, if we combine the first steps of \cite{veliz2011reduction} and \cite{saadatpour2013reduction}, we may achieve interesting reductions which exactly preserve all asymptotic dynamics.
	
	% Moreover, the absorption of single mediator node should not be applied to synchronous BNs since it also affects the number and the sizes of cyclic attractors \comg{example}.
	
	The first steps of \cite{veliz2011reduction,saadatpour2013reduction} affect only a BN property called \emph{stability}. Stability is the ability of a BN to end up to the same attractor when starting from slightly different initial conditions. % i.e. initial conditions that differ, for example, only in one bit. 
	In \cite{bilke2001stability}, the authors introduced the decimation procedure -a reduction procedure for synchronous BNs- to discuss how it affects stability. The crucial difference between decimation procedure and BBE-reduction is that the first was invented to study stability whereas the latter was invented to degrade state space explosion. The decimation procedure is summarized by the following four steps: (i) remove from every update functions the inputs that it does not depend on,
	(ii) find the constant value for variables with no inputs, (iii) propagate the constant values to other update functions and remove this variable from the system, and (iv) if a variable has become constant, repeat from step (i). The study also refers to leaf variables because their presence does not play any role in the asymptotic dynamics of a BN. However, both leaf and fixed-valued variables affect stability. Overall, the decimation procedure exactly preserves the asymptotic dynamics of the original model since it throws out only variables considered as asymptotically irrelevant.

	\section{Conclusion}\label{sec6}
	
	We introduced an automatic reduction technique for synchronous Boolean Networks which preserves  dynamics of interest. The modeller gets a reduced BN based on requirements expressed as an initial partition of variables.
	The reduced BN can recover a pure part of the original state space and its trajectories established by the reduction isomorphism. Notably, we draw connections between the STG of the original and that of the reduced BN through a rigorous mathematical framework. The dynamics preserved are those wherein collapsed variables have equal values.
	
	We used our reduction technique to speed-up attractor identification. Despite that the length of the preserved attractors is consistent in the reduced model, some of them may get lost. In the future, we plan to study classes of initial partitions that preserve all attractors. %This is feasible if we specify initial partitions that do not degenerate feedback loops as proposed in \cite{grieco2013integrative,rodriguez2019cooperation}.
	We have shown the analysis speed-ups obtained for attractor identification as implemented in the tool BoolSim~\cite{BoolsimSoftw}. In the future we plan to perform a similar analysis on a recent attractor identification approach from~\cite{dubrova2011sat}.
	
	Our method was implemented in ERODE~\cite{cardelli2017erode}, a freely available tool for reducing biological systems. Related \emph{quantitative} techniques offered by ERODE have been recently validated on a large database of biological models~\cite{DBLP:conf/cmsb/Perez-VeronaTV19,cmsb2019tcs,DBLP:journals/corr/abs-2101-03342}. In the future we plan to extend this analysis considering also BBE. We also plan to investigate whether BBE can be extended  in order to be able to compare different models as done for its quantitative counterparts~\cite{DBLP:conf/lics/CardelliTTV16,DBLP:journals/tcs/CardelliTTV19}. 
	
	Our method could be combined with most of the existing methods found in literature. Our prototype toolchain consists of several tools from the COLOMOTO interoperability initiative. We aim to incorporate our toolchain into the COLOMOTO Interactive Notebook~\cite{10.3389/fphys.2018.00680}, a unified environment to edit, execute, share, and reproduce analyses of qualitative models of biological networks.
	
	Multivalued BNs, i.e. whose variables can take more than two activation values, are currently supported only via a \emph{booleanization} technique~\cite{delaplace2020bisimilar,chaouiya2013sbml} that might hamper the interpretability of the reduced model. In future work we plan to generalize BBE to support directly multivalued networks. 
	%Furthermore, in a similar perspective, one could define conjunctive (disjunctive) equivalence where the modeller is especially interested in the conjunction (disjunction) of specific variables.

	%\paragraph*{Acknowledgement}
	%This work was supported by the Independent Research Fund Denmark Research Project 9040-00224B REDUCTO. %: A novel approach to the reduction of Boolean Networks.

	\bibliographystyle{splncs04}
	\bibliography{ms}

\newpage

\appendix

\section{Proofs}\label{sec:proofs}

\begin{proof}[Proof of Theorem~\ref{th:maximal}]
Let $X_{R_1}$, $X_{R_2}$ be two BBE partitions that refine some other partition $X_I$ that is not necessarily a BBE. We start by noting that $R = (R_1 \cup R_2)^\ast \subseteq I$ because $R_1, R_2 \subseteq I$, where asterisk denotes the transitive closure, while $R_1, R_2$ and $I$ are equivalence relations underlying $X_{R_1}, X_{R_2}$ and $X_I$, respectively. Hence, $X_R$ is a refinement of $X_I$. We next show that $X_R$ is a BBE partition. To this end, fix some $\textbf{s} \in \mathbb{B}^n$ that is constant on $R$. Since $R_i \subseteq R$, this implies that $\textbf{s}$ is constant on $R_i$ which, in virtue of $X_{R_i}$ being a BBE, implies that $F(\textbf{s}) \in \mathbb{B}^n$ is constant on $R_i$. This implies that $F(\textbf{s}) \in \mathbb{B}^n$ is constant on $R = (R_1 \cup R_2)^\ast$, thus showing that $X_R$ is indeed a BBE partition. The overall claim follows by noting that the finiteness of $X$ implies that there are finitely many BBE partitions $X_{R_i}$ that refine any given partition $X_I$ of $X$.
\end{proof}

%			\For{$\forall C \in H$}
%			{$C_0 = \{x_i \in C : f_i(\mathbf{s}) = 0\}$\;
%				$C_1 = \{x_i \in C : f_i(\mathbf{s}) = 1\}$\;
%				$H' = H' \cup C_1 \cup C_0$\;
%				%$R \leftarrow \{(x_i,x_j):x_i,x_j \in H_l$ and $f_i(s)=f_j(s)\} $\;
%				%$H^\prime \leftarrow H^\prime \cup (H_l/R)$\;
%			}
%			

\begin{proof}[Proof of Theorem~\ref{th:computesmaximal}]
Assume that $G'$ denotes the coarsest BBE partition that refines some given partition $G$. Set $H_0 := G$ and define for all $k \geq 0$
\begin{align*}
H_{k+1} := \big( \{ C_0 \mid C \in H_k \} \cup \{ C_1 \mid C \in H_k \} \big) \setminus \{\emptyset\},
\end{align*}
where $C_0$ and $C_1$ are as in Algorithm~\ref{algorithm}. Then, a proof by induction over $k \geq 1$ shows that (a) $G'$ is a refinement of $H_k$ and (b) $H_k$ is a refinement of $H_{k-1}$, for all $k \geq 1$. Since $G'$ is a refinement of any $H_k$, it holds that $G' = H_k$ if $H_k$ is a BBE partition. Since $X$ is finite, b) allows us to fix the smallest $k \geq 1$ such that $H_k = H_{k-1}$. This, in turn, implies that $H_{k-1}$ is a BBE.
\end{proof}

\begin{proof}[Proof of Lemma~\ref{lemma:constantAttractors}]
The fact that $A \cap S_{\mid R} \neq \emptyset$ implies that there is at least one state $\textbf{s}\in A$ that is constant on $R$, i.e., $\textbf{s} \in A \cap S_{\mid R}$. For any such state $\textbf{s}$, by the properties of BBE we have that any state $\textbf{t}$ such that $\textbf{s} \rightarrow^+ \textbf{t}$ is also constant on $R$. Actually,  it is trivial to show that $A = \{ \textbf{t} \mid \textbf{s} \rightarrow^+ \textbf{t} \}$. It immediately follows that $A \subseteq S_{\mid R}$.
\end{proof}

\begin{proof}[Proof of Lemma~\ref{lemma:isomorphism}]
Follows readily from the definition of a BBE and $m_R$.
%	\comal{to be done, but one has basically to refer to the definition of BBE reduction and of $m_R$.}
\end{proof}

\begin{proof}[Proof of Theorem~\ref{th:attractorPreservation}]
	The theorem trivially follows from Lemmas~\ref{lemma:constantAttractors} and~\ref{lemma:isomorphism}.
\end{proof}

\section{Table of large-scale validation}\label{largescaletable}%\label{appA}
We provide the table referenced in the Section \ref{sec:42} on large-scale validation of BBE. The table contains the results of BBE reduction on all the models from the GINsim repository. The first column contains the model identifier (MI). The second column contains %the name of the model with its corresponding literature 
the url to download the model and the third column displays the number of variables in the original BN. In the case of multivalued networks, the column contains the number of variables after booleanization. We denote with $N_m,$ $N_i$ the number of variables of the maximal and the ID reduced BN in the fifth and sixth column respectively. Note that when a BN has no input variables $N_i$ and $N_m$ coincide. The last two columns display the reduction ratios $r_i=N_i/N$, $r_m=N_m/N$ where $N$ is the number of variables in the original BN.

\newcolumntype{H}{>{\setbox0=\hbox\bgroup}c<{\egroup}@{}}
%\begin{wraptable}{r}{10cm}
\begin{table}
	\scriptsize
	\begin{center}
		\scalebox{1.1}{
%				\begin{tabular}{c | c |  l | r | r | r | c | c} 
			\begin{tabular}{c| c | H  r | r | r | c | c} 
				
				\hline
				\emph{MI}  & \emph{GINsim repository URI} & \emph{Model name} %& \emph{Ref}                                                                          
				& $N$     
				& $N_i$       & $N_m$    & $r_i$             & $r_m$             \\ \hline
				B1&\url{http://ginsim.org/node/225}  & TCR and TLR5 %merged Boolean model 
				%& \cite{rodriguez2019cooperation}
				& 128     & 107         & 95       & 0.836         & 0.742         \\ \hline
				B2& \url{http://ginsim.org/node/225} & TCR  
				%& \cite{rodriguez2019cooperation}                                                            
				& 110     & 103         & 91       & 0.936 & 0.827 \\ \hline
				B3& \url{http://ginsim.org/node/87} & T-LGL %& \cite{zhang2008network} 
				& 60      & 57          & 52       & 0.95              & 0.867 \\ \hline
				B4& \url{http://ginsim.org/node/173} & MAPK network 
				%& \cite{grieco2013integrative}                                                      
				& 53      & 46          & 39       & 0.868 & 0.736 \\ \hline
				B5& \url{http://ginsim.org/node/225} & TLR5 
				%& \cite{rodriguez2019cooperation}                                                           
				& 42      & 37          & 29       & 0.881 & 0.690  \\ \hline
				B6& \url{http://ginsim.org/node/78} & TCR signalisation 
				%& \cite{klamt2006methodology}                                                  
				& 40      & 31          & 29       & 0.775             & 0.725             \\ \hline
				B7& \url{http://ginsim.org/node/227} & Cell-Fate Decision in Response to Death Receptor Engagement (i) 
				%& \cite{calzone2010mathematical} 
				& 33 & 27     & 25       & 0.818 & 0.758 \\ \hline
				B8& \url{http://ginsim.org/node/191} & Molecular Pathways Enabling Tumour Cell Invasion and Migration%& \cite{cohen2015mathematical}     
				& 32      & 32          & 31       & 1                 & 0.969           \\ \hline
				B9& \url{http://ginsim.org/node/227} & Cell-Fate Decision in Response to Death Receptor Engagement(ii) 
				%& \cite{calzone2010mathematical} 
				& 28      & 25          & 20       & 0.893 & 0.714 \\ \hline
				B10& \url{http://ginsim.org/node/97} & Drosophila Wg Signaling pathway %& \cite{mbodj2013logical}                                        
				& 26      & 23          & 4        & 0.885 & 0.154 \\ \hline
				B11& \url{http://ginsim.org/node/144} & Drosophila Hh Signaling pathway %& \cite{mbodj2013logical}                                        
				& 24      & 23          & 9        & 0.958 & 0.375             \\ \hline
				B12& \url{http://ginsim.org/node/126} & Drosophila SPATZLE Processing pathway %& \cite{mbodj2013logical}                                  
				& 24      & 21          & 4        & 0.875             & 0.167 \\ \hline
				B13& \url{http://ginsim.org/node/102} & Drosophila FGF Signaling pathway %& \cite{mbodj2013logical}                                       
				& 23      & 22          & 8        & 0.957 & 0.348 \\ \hline
				B14& \url{http://ginsim.org/node/39} & ERBB receptor-regulated G1/S transition model %& \cite{sahin2009modeling}                         
				& 20      & 15          & 13       & 0.75              & 0.65              \\ \hline
				B15& \url{http://ginsim.org/node/160} & Drosophila VEGF Signaling pathway %& \cite{mbodj2013logical}                                      
				& 18      & 18          & 8        & 1                 & 0.444 \\ \hline
				B16& \url{http://ginsim.org/node/35} & Budding yeast cell cycle %& \cite{irons2009logical}                                               
				& 18      & \multicolumn{2}{c|}{17} & \multicolumn{2}{c}{0.944}                   \\ \hline
				B17& \url{http://ginsim.org/node/31} & Drosophila cell cycle %& \cite{faure2009logical}                                                   
				& 14      & 14          & 12       & 1                 & 0.857 \\ \hline
				B18& \url{http://ginsim.org/node/152} & Drosophila Toll Signaling pathway %& \cite{mbodj2013logical}                                      
				& 11      & 10          & 9        & 0.909 & 0.818 \\ 
				\hline
				B19& \url{http://ginsim.org/node/37} & Fission Yeast Cell Cycle %& \cite{faure2009logical}                                               
				& 10      & 9           & 9        & 0.9               & 0.9               \\ \hline
				B20& \url{http://ginsim.org/node/69} & DV boundary formation of the Wing imaginal disc %& \cite{gonzalez2006dynamical}                   
				& 10      & 8           & 8        & 0.8               & 0.8               \\ \hline
				B21& \url{http://ginsim.org/model/C\_crescentus} & Asymmetric Cell Division in Caulobacter Crescentus g2b %& \cite{sanchez2017modeling}              
				& 9       & \multicolumn{2}{c|}{7}  & \multicolumn{2}{c}{0.778}                   \\ 
				\hline
				B22& \url{http://ginsim.org/node/21} & Budding yeast cell cycle %& \cite{faure2009logical}                                               
				& 9       & \multicolumn{2}{c|}{7}  & \multicolumn{2}{c}{0.778}                  \\ \hline
				B23& \url{http://ginsim.org/node/214} & miR-9 and timing of neurogenesis %& \cite{coolen2012mir}                                          
				& 6       & \multicolumn{2}{c|}{2}  & \multicolumn{2}{c}{0.333}                    \\ \hline
				B24& \url{http://ginsim.org/model/C\_crescentus} & Asymmetric Cell Division in Caulobacter Crescentus g2a %& \cite{sanchez2017modeling}              
				& 5       & \multicolumn{2}{c|}{1}  & \multicolumn{2}{c}{0.2}                 \\ \hline
				%\multicolumn{5}{|c|}{Average}                                                & 0.834 & 0.66 \\ \hline
				%  & Name of Multivalued Network and literature    &   \multicolumn{5}{c|}{}             \\ \hline
				M1&\url{http://ginsim.org/model/tcell-checkpoint-inhibitors-tcla4-pd1} & Hernandez\_TcellCheckPoints\_13april2020 %& \cite{hernandez2020computational}  
				& 218 & 201        & 136       & 0.922 & 0.623 \\ \hline
				M2&\url{http://ginsim.org/node/229} & TCR\_REDOX\_METABOLISM\_2019\_07\_26  %& \cite{sanchez2019contribution}             
				& 133 & 126        & 122       & 0.947 & 0.917 \\ \hline
				M3&\url{http://ginsim.org/node/178} & vcpwt23h  %& \cite{weinstein2013network} 
				& 107 & 93         & 59        & 0.869 & 0.551 \\ \hline
				M4&\url{http://ginsim.org/node/185} & Frontiers\_Th\_Fullannotated   %& \cite{abou2015model}            
				& 103 & 97         & 52        & 0.941 & 0.505 \\ \hline
				M5& \url{http://ginsim.org/model/SP} & SP\_6cells %& \cite{sanchez2002segmenting}  
				& 102 & \multicolumn{2}{|c}{16} & \multicolumn{2}{|c}{0.157}                  \\ \hline
				M6& \url{http://ginsim.org/node/194} & Flobak\_FullModel\_S2\_Dataset %& \cite{flobak2015discovery}                    
				& 83  & \multicolumn{2}{|c|}{79} &\multicolumn{2}{|c}{ 0.951}                   \\ \hline
				M7& \url{http://ginsim.org/node/79} & Th\_differentiation\_full\_annotated\_model %& \cite{naldi2010diversity}        
				& 71  & 69         & 42        & 0.972 & 0.592 \\ \hline
				M8& \url{http://ginsim.org/node/234}  & Cacace\_Tdev\_2nov2019 %& \cite{cacace2020logical} 
				& 61  & 60         & 58        & 0.983 & 0.95 \\ \hline
				M9& \url{http://ginsim.org/model/drosophila\_mesoderm} & DrosoMesoLogModel %& \cite{mbodj2016qualitative}  
				& 57  & 57         & 11        & 1                 & 0.192 \\ \hline
				M10& \url{http://ginsim.org/node/69} & ap\_boundary  %& \cite{gonzalez2008logical}  
				& 56  & 56         & 50        & 1                 & 0.893 \\ \hline
				M11& \url{http://ginsim.org/model/EMT\_Selvaggio\_etal} & Selvaggio\_etal\_2020 %& \cite{selvaggio2020hybrid} 
				& 56  & 56         & 43        & 1                 & 0.768 \\ \hline
				M12& \url{http://ginsim.org/node/229} & TCR\_REDOX\_METABOLISM\_2019\_07\_26\_reduced\_0 %& \cite{sanchez2019contribution}   
				& 53  & 51         & 50        & 0.962 & 0.943 \\ \hline
				M13& \url{http://ginsim.org/node/21}  & coupled\_budding\_yeast\_CC %& \cite{faure2009modular}     
				& 50  & 49         & 41        & 0.98              & 0.82              \\ \hline
				M14& \url{http://ginsim.org/node/180} & Mast cell activation  %& \cite{niarakis2014computational}          
				& 48 &\multicolumn{2}{|c|} {35}& \multicolumn{2}{|c}{0.729}  \\ \hline
				M15& \url{http://ginsim.org/node/25} & core\_engine\_budding\_yeast\_CC %& \cite{faure2009modular}
				& 39  & 37         & 31        & 0.948 & 0.794 \\ \hline
				M16& \url{http://ginsim.org/model/sex\_determination\_chicken} & primary\_sex\_determination\_2%& \cite{sanchez2016primary}
				& 37  & 37         & 14        & 1                 & 0.378 \\ \hline
				M17& \url{http://ginsim.org/node/79} & Th\_differentiation\_reduced\_model %& \cite{naldi2010diversity}             
				& 36  & 35         & 21        & 0.972 & 0.583 \\ \hline
				M18& \url{http://ginsim.org/node/188} & Bladder\_Model %& \cite{remy2015modeling}                                  
				& 35  & 34         & 28        & 0.971 & 0.8               \\ \hline
				M19& \url{http://ginsim.org/node/216} & Collombet\_model\_Bcell\_Macrophages\_PNAS\_170215 %& \cite{collombet2017logical} 
				& 34  & 34         & 33        & 1                 & 0.971 \\ \hline
				M20& \url{http://ginsim.org/node/96} & EGF\_\_Pathway\_12Jun2013\_0 %& \cite{mbodj2013logical}
				& 34  & 32         & 15        & 0.941 & 0.441 \\ \hline
				M21& \url{http://ginsim.org/node/183} & Senescence\_Onset\_G1\_S %& \cite{mombach2014modelling}  
				& 30  & 30 & 14        & 1    & 0.467 \\ \hline
				M22& \url{http://ginsim.org/model/eggshell\_patterning} & mechanistic\_cellular %& \cite{faure2014discrete} 
				& 24  & 23         & 12        & 0.958 & 0.5               \\ \hline
				M23& \url{http://ginsim.org/node/41} & Th\_17 %& \cite{mendoza2006network} 
				& 21  & 21         & 18        & 1                 & 0.857 \\ \hline
				M24& \url{http://ginsim.org/model/sex\_determination\_mammals} & primary\_sex\_determination\_1%& \cite{sanchez2016primary}
				& 19  & 19         & 12        & 1 & 0.632 \\ \hline
				M25& \url{http://ginsim.org/node/109} & JakStat\_\_Pathway\_12Jun2013% & \cite{mbodj2013logical} 
				& 19  & 19         & 7         & 1 & 0.368 \\ \hline
				M26& \url{http://ginsim.org/model/SP} & SP\_1cell %& \cite{sanchez2002segmenting} 
				& 19  & 18         & 17        & 0.947 & 0.895 \\ \hline
				M27& \url{http://ginsim.org/node/89} & Dpp\_\_Pathway\_11Jun2013 %& \cite{mbodj2013logical}
				& 18  & 18         & 10        & 1                 & 0.556 \\ \hline
				M28& \url{http://ginsim.org/node/29} & exit  %& \cite{faure2009modular}  
				& 16  & 16         & 13        & 1                 & 0.812            \\ \hline
				M29& \url{http://ginsim.org/model/sex\_determination\_chicken} & full\_network % & \cite{sanchez2018logical}   
				& 15  & 15         & 13        & 1                 & 0.866 \\ \hline
				M30& \url{http://ginsim.org/node/194} & Flobak\_ReducedModel\_S3\_Dataset  %& \cite{flobak2015discovery}               
				& 14  & 14         & 13        & 1                 & 0.928 \\ \hline
				M31& \url{http://ginsim.org/node/26} & MCP\_budding\_yeast\_CC %& \cite{faure2009modular}     
				& 12  & 12         & 8         & 1                 & 0.667 \\ \hline
				M32& \url{http://ginsim.org/node/115} & Notch\_\_Pathway\_12Jun2013 %& \cite{faure2009modular}                       
				& 16  & 16         & 5         & 1                 & 0.3125            \\ \hline
				M33& \url{http://ginsim.org/node/220} & reduced\_network\_0   %& \cite{sanchez2018logical}  
				& 10  & 10         & 8         & 1                 & 0.8               \\ \hline
				M34& \url{http://ginsim.org/model/eggshell\_patterning} & phenomenological\_cellular %& \cite{faure2014discrete} 
				& 8   & 8          & 2         & 1                 & 0.25              \\ \hline
				M35& \url{http://ginsim.org/node/82} & gapB %& \cite{sanchez2001logical}      
				& 7   & 7          & 6         & 1                 & 0.857 \\ \hline
				M36& \url{http://ginsim.org/node/82} & gapD %& \cite{sanchez2001logical}       
				& 7   & 7          & 6         & 1                 & 0.857 \\ \hline
				M37& \url{http://ginsim.org/node/50} & Trp\_reg %& \cite{simao2005qualitative}
				& 6   & 6          & 5         & 1                 & 0.833 \\ 
				%\hline
			%	\multicolumn{5}{c|}{\emph{Average}}                                                             & 0.850 & 0.672 
			\end{tabular}
		}
	\end{center}	
	\caption{Application of BBE to BNs from the GINsim model repository.}
	\label{table with results} 
\end{table}

\section{Refined initial partitions for the  selected case studies}\label{sec:appendix_casestudies}

In Section~\ref{casestudies}, we studied how BBE affects attractor analysis of three selected case studies.
It is remarkable that attractor identification was infeasible for the largest TCR-TLR5 BN, whereas we identified its attractors in 30 seconds in its maximal reduction.
%
%\paragraph{Refined Reduced Models-Hypothesis}
However, the attractors identified %in the maximal reduced BN 
may not be all the attractors of interest for the BN. Our crucial hypothesis is that one can specify alternative initial partitions that preserve more or discard some irrelevant attractors. 
%We anticipate that initializing Algorithm~\ref{algorithm} with initial partition different from that of the maximal reduction, we could obtain a reduced BN that preserves more attractors. 
We also expect that the reduction ratio of these alternative initial partitions lies between that of the ID and that of the maximal reduction ($r_m$, $r_i$).

\paragraph{Configuration} 
In Sections~\ref{sec:merged}~,~\ref{sec:MAPK}~and~\ref{sec:TLGL}~, we provide a detailed description of the initial partitions that lead to the \emph{refined reduced models}. 

\paragraph{Results}

The results of the refined MAPK, the refined merged TCR-TLR, and the refined T-LGL reduced models are summarized in the following Table~\ref{rt1}~.  We present the number of variables (\emph{Size}), the number of \emph{Attractors}, and the time needed for reduction (\emph{Reduction (s)}) and attractor identification (\emph{Analysis~(s)}).

\begin{table}[h]
	\centering
	\scalebox{0.9}{
		\begin{tabular}{c|ccc|cccc}
			\hline
			\emph{Model} &   \multicolumn{3}{c|}{\emph{Original model}}     &  \multicolumn{4}{c}{\emph{Refined 
					Reduced model}} \\
			\hline
			& \emph{Size}& \emph{Attractors} & \emph{Analysis (s)}&\emph{Reduction (s)} &\emph{Size}& \emph{Attractors} & \emph{Analysis (s)}\\
			%\hline
			%TCRsig&40& 8 & 0,148s & 0,673s&31  &8 & 0,109s\\
			\hline
			MAPK Network	& 53 &40 &	16.501	& 1.202      &42  &40  & 12.115\\
			\hline
			T-LGL &60 &264 &123.431 &1,816 &56 &120 &55.049\\
			\hline
			TCR-TLR merged  &128 &\multicolumn{2}{c|}{---\emph{Time Out}---} &2.096 & 98 & 8&9349.577\\
		\end{tabular}
	}
	\caption{The results of 3 case studies for the original and the refined reduced BNs}
	\label{rt1}
\end{table}

The refined reduced MAPK network consists of 42 variables but preserves all attractors in the original model. %, asymptotic dynamics. This is an immediate results from Theorems~\ref{theorem3}~,~\ref{theorem4},~Lemma~\ref{lemma2} and the fact that we identified 40 attractors in the reduced BN-as many as in the original. 
Notably, the reduced model has $\mathit{ 80\%}$ the size of the original. The refined reduced T-LGL results from the original after specifying two input variables in the same block of the initial partition. The merging of these two input variables is an immediate result of~\cite{zhang2008network} and discards 144 attractors which are irrelevant for their analysis. Last but not least, the refined reduced TCR-TLR merged has 98 variables and 8 attractors-more than the maximal reduction of Table~\ref{rt}~. Note again that the attractor identification in the original BN is intractable.

\paragraph{Interpretation} Overall, Table~\ref{rt1} illustrates the possibility of analyzing large BNs by defining alternative initial partitions than the two considered in Section~\ref{sec:42}. Alternative reductions may provide fruitful insights and identify crucial properties of the underlying system. The size of the refined reduced model lies between the size of the input-distinguished and the maximal reduced model in all three models. 

The initialization of Algorithm \ref{algorithm} provides also a framework to specify desires and limitations. Desires refer to the preserved properties with respect to the original model whereas limitations refer to variable perturbation. If such a variable get merged then its perturbation will indicate subsequent perturbation to all the variables that belong to its class. Consequently, variables that are amenable to perturbation, should be kept in singleton blocks of the initial partition.
To this end, we can construct empirical initial partitions that (i) preserve attractors, (ii) discard some of them, or (iii) isolate in singleton blocks variables which are amenable to perturbation.

%\begin{enumerate}

%\item Choose a model: at this point in time, it might be convenient choosing one or two models from Andrea's paper E.g, TCR signalisation.

%\item Provide a description (what it is about) and a picture of the considered BN

%\item Analysis of the model (before reductions)
%\begin{enumerate}
%\item Compute the analysis that people did on this model in other papers. 
%\item Or, think about some analysis to do on the model (e.g., STG generation + steady states)
%\end{enumerate}

%\item Reduce the model explaining why we are using a reasonable initial partition (it allows us to do the same experiments)
%begin{enumerate}
%\item Is the reduced model 'easier/simpler' to read? Do we get interesting information by the reduction itself?
%\item Explain what is preserved and what is lost in the STG (e.g. attractors - in connection to the analysis done above)

%\end{enumerate}

%\item Repeat analysis of step 3 on the reduced model (after BBE-reduction). This should  show that BBE-reductions help in reducing analysis time. 

%\end{enumerate}

\subsection{T Cell and Toll-like Receptor (TCR-TLR) merged signalling BN}\label{sec:merged}

In this Section, we exploit the results from the maximal and the ID reduction to obtain two refined reduced BNs of the TCR-TLR merged BN. This BN refers to the T cell receptors and their responsibility for the activation of T cell (\cite{rodriguez2019cooperation}, Fig.~\ref{merged}). The authors generated logical models for the TCR and the TLR5 signalling pathways, and merged them by considering their cross interactions. The original model contains 128 variables, fact that renders its analysis intractable. In order to experimentally validate the correctness of their new merged BN, they considered asynchronous update schema and performed reduction with absorption~\cite{naldi2011dynamically}. Absorption does not guarantee
%Furthermore, they claim that, if feedback loops are not eliminated, absorption preserves all attractors. 
preservation of all asymptotic dynamics. It has been proven~\cite{veliz2011reduction,naldi2011dynamically,saadatpour2013reduction} that preserves only steady states and may cause spurious cyclic attractors when applied to asynchronous dynamics. When applied to synchronous dynamics, this method may also degenerate cyclic attractors. BBE-reduction maintain the lengths of the preserved attractors according to Theorem~\ref{th:attractorPreservation}. 
%If we specify initial partitions that do not eliminate feedback loops \comg{how you do this?}\comal{ok, but how do we ensure feedback loops are not eliminated?}, we can also preserve the number of attractors. Hence, we propose a more automated workflow: (i) apply automatic BBE-reduction, and if further reduction is needed (ii) reduce manually by iteratively applying absorption.

\paragraph{ID reduction}

The application of ID reduction to the merged model resulted in 10 equivalence classes displayed in Fig.~\ref{idequivalentclasses}~. We also display them in Fig.~\ref{merged} with different colors for each class: Backward equivalent variables are represented with colored boxes, and colored boxes that belong to the same equivalence class have the same background. Variables in white background belong to singleton classes. The ID reduced BN is still huge (116 variables) and the attractor identification is intractable.

\begin{figure}[h]
	\centering
	\scalebox{0.9}{
		\begin{tabular}{cc}
			$\mathit{\{IRAK4}$, $\mathit{PIK3AP1}\}$&  $\mathit{\{GRAP2}$, $\mathit{MAP4K1}\}$\\
			$\mathit{\{ TICAM1}$, $\mathit{MyD88}\}$& $\mathit{\{MKNK1}$, $\mathit{RPS6KA5}\}$\\ 
			$\mathit{\{ Foxo1}$, $\mathit{BAD}$, $\mathit{GSK3B}$, $\mathit{CDKN1A}\}$& $\mathit{\{MAPKAPK2}$, $\mathit{mTOR}\}$\\
			$\mathit{\{MAP2K3}$, $\mathit{MAP2K7}\}$& $\mathit{\{Camk2}$, $\mathit{Camk4}\}$\\
			$\mathit{\{Cyc1}$, $\mathit{CTNNB1}\}$& $\mathit{\{DAG}$, $\mathit{IP3}\}$\\
		\end{tabular}
	}
	\caption{The equivalence classes of the input-distinguished reduction}
	\label{idequivalentclasses}
\end{figure}
%\begin{itemize}
%		\item $\mathit{\{IRAK4}$, $\mathit{PIK3AP1}\}$
%		\item $\mathit{\{ TICAM1}$, $\mathit{MyD88}\}$
%		\item $\mathit{\{ Foxo1}$, $\mathit{BAD}$, $\mathit{GSK3B}$, $\mathit{CDKN1A}\}$
%		\item $\mathit{\{MAP2K3}$, $\mathit{MAP2K7}\}$
%		\item $\mathit{\{Cyc1}$, $\mathit{CTNNB1}\}$
%		\item $\mathit{\{GRAP2}$, $\mathit{MAP4K1}\}$
%		\item $\mathit{\{MAPKAPK2}$, $\mathit{mTOR}\}$
%		\item $\mathit{\{MKNK1}$, $\mathit{RPS6KA5}\}$
%		\item $\mathit{\{Camk2}$, $\mathit{Camk4}\}$
%		\item $\mathit{\{DAG}$, $\mathit{IP3}\}$
%\end{itemize}
%$\mathit{\{IRAK4}$, $\mathit{PIK3AP1}\}$,  $\mathit{\{ TICAM1}$, $\mathit{MyD88}\}$, $\mathit{\{ Foxo1}$, $\mathit{BAD}$, $\mathit{GSK3B}$, $\mathit{CDKN1A}\}$,  $\mathit{\{MAP2K3}$, $\mathit{MAP2K7}\}$, $\mathit{\{Cyc1}$, $\mathit{CTNNB1}\}$, $\mathit{\{GRAP2}$, $\mathit{MAP4K1}\}$, $\mathit{\{MAPKAPK2}$, $\mathit{mTOR}\}$,  $\mathit{\{MKNK1}$, $\mathit{RPS6KA5}\}$, $\mathit{\{Camk2}$, $\mathit{Camk4}\}$, $\mathit{\{DAG}$, $\mathit{IP3}\}$. 
%Boolsim could not compute the attractors of the original and the reduced model in 24 hours, when we terminated the analysis. 

\begin{figure}%[ht]
	\centering
	\begin{tabular}{cc}
		\includegraphics[scale=0.25]{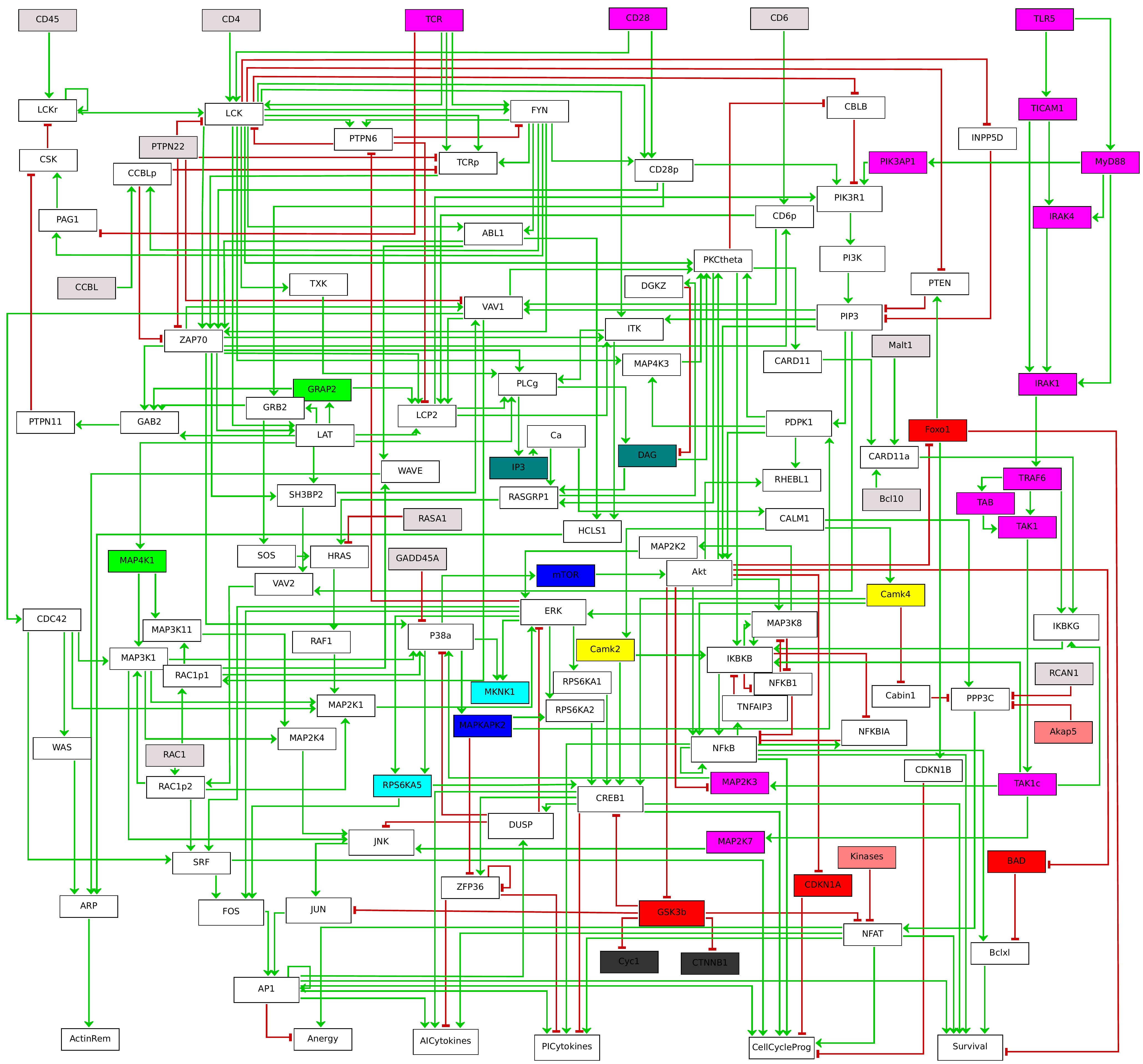}
		\\
		\hline
		\includegraphics[scale=0.25]{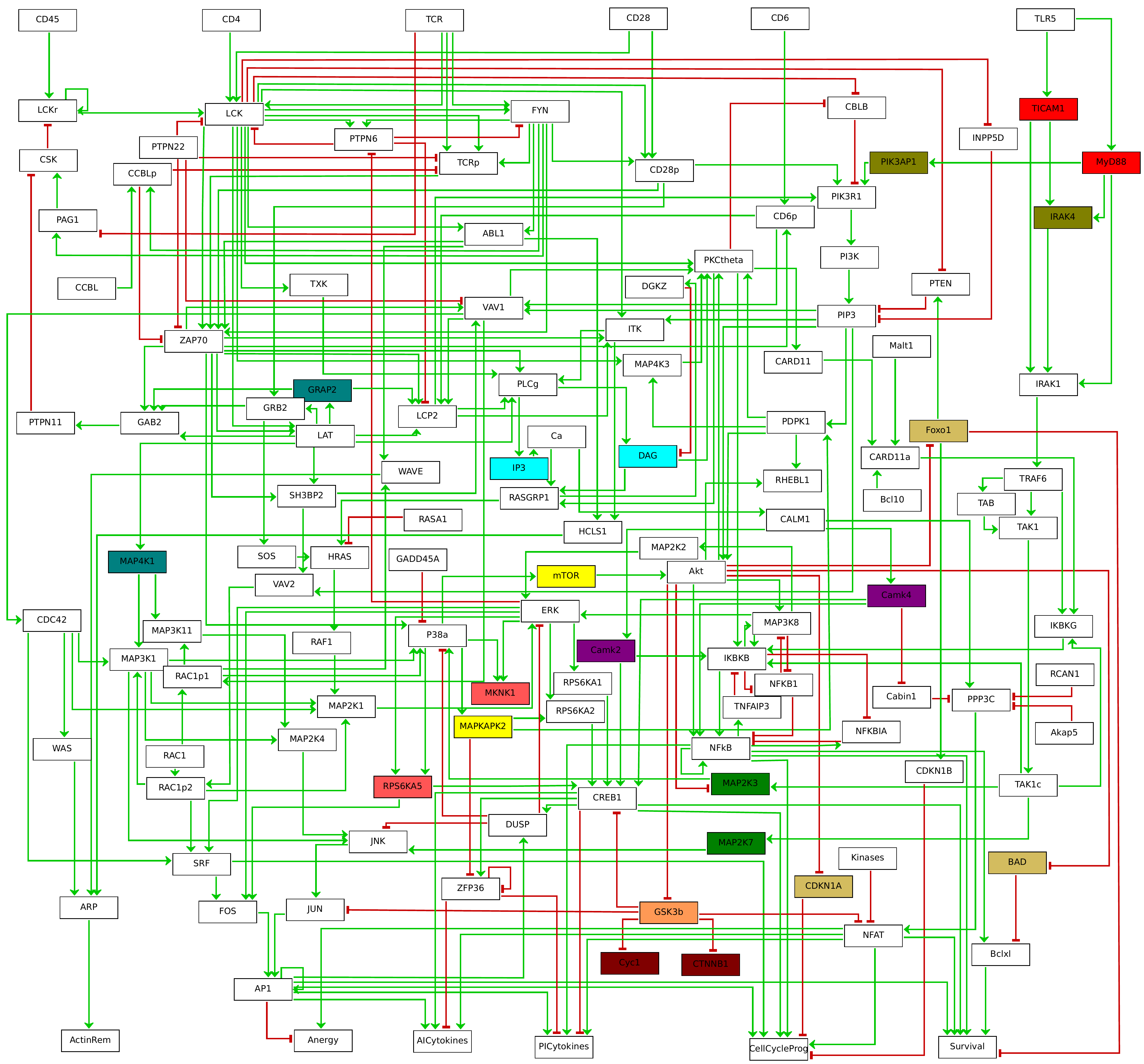}
	\end{tabular}
	\caption{Up: The  TCR-TLR merged BN with input-distinguishing BBE-variables having the same background color.Bottom: The  TCR-TLR merged BN with maximal-reduction BBE-variables having the same background color. Variables with white backround belong to singleton classes.}
	\label{merged}
\end{figure}

\paragraph{Maximal reduction}

However, the maximal reduced BN contains $\mathit{95}$ variables and the attractor identification is feasible in $\mathit{29.336}$ seconds. Fig.~\ref{merged} displays the following BBE-equivalence classes with different colors:

\begin{figure}[h]
	\centering
	\scalebox{0.8}{
		\begin{tabular}{cc}
			*$\mathit{\{CD45}$, $\mathit{CD4}$, $\mathit{CD6}$, $\mathit{CCBL}$,  &***$\mathit{\{TCR}$, $\mathit{CD28}$, $\mathit{TLR5}$, $\mathit{TICAM1}$, $\mathit{PIK3AP1}$, \\
			$\mathit{PTPN22}$, $\mathit{Malt1}$, $\mathit{Bcl10}$, $\mathit{RCAN1}$, & $\mathit{MyD88}$, $\mathit{IRAK1}$ $ \mathit{TRAF6}$, $\mathit{TAB}$, $\mathit{TAK1}$,\\
			$\mathit{RAC1}$, $\mathit{RASA1}$, $\mathit{GDD45A}\}$ & $\mathit{TAK1c}$, $\mathit{IRAK4}$, $\mathit{MAP2K3}$, $\mathit{MAP2K7}\}$ \\
			$\mathit{\{MAPKAPK2}$, $\mathit{mTOR}\}$&** $\mathit{\{Kinases}$, $\mathit{Akap5}\}$\\ 
			$\mathit{\{MKNK1}$, $\mathit{RPS6KA5}\}$& $\mathit{\{ Foxo1}$, $\mathit{BAD}$, $\mathit{GSK3B}$, $\mathit{CDKN1A}\}$\\
			$\mathit{\{DAG}$, $\mathit{IP3}\}$& $\mathit{\{GRAP2}$,  $\mathit{MAP4K1}\}$\\
			$\mathit{\{Camk2}$, $\mathit{Camk4}\}$& $\mathit{\{Cyc1}$, $\mathit{CTNNB1}\}$\\
		\end{tabular}
	}
	\caption{The equivalence classes of the maximal reduction}
	\label{maxequivalentclasses}
\end{figure}

The maximal reduction splits input variables into 3 classes: The first class (* in Fig.~\ref{maxequivalentclasses}) %$\mathit{\{CD45}$, $\mathit{CD4}$, $\mathit{CD6}$, $\mathit{CCBL}$, $\mathit{PTPN22}$, $\mathit{Malt1}$, $\mathit{Bcl10}$, $\mathit{RCAN1}$, $\mathit{RAC1}$, $\mathit{RASA1}$, $\mathit{GDD45A}\}$ 
contains all variables with stable update function which equals to \emph{true}. The second class (**)  %$\mathit{\{Kinases}$, $\mathit{Akap5}\}$ 
contains all variables with stable update functions that equals to \emph{false}. The third class (***) % $\mathit{\{TCR}$, $\mathit{CD28}$, $\mathit{TLR5}$, $\mathit{TICAM1}$, $\mathit{PIK3AP1}$, $\mathit{MyD88}$, $\mathit{IRAK1}$, $\mathit{TRAF6}$, $\mathit{TAB}$, $\mathit{TAK1}$, $\mathit{TAK1c}$, $\mathit{MAP2K3}$, $\mathit{MAP2K7}\}$ 
contains all variables with identity update function ($\mathit{TCR}$, $\mathit{CD28}$, $\mathit{TLR5}$) and all variables BBE-equivalent variables with them. 

%Based on observations of the equivalence classes between the maximal and the input distinguished reduction, we

\paragraph{Refined Reduction}

We now consider two alternative initial partitions, initialize Algorithm~\ref{algorithm} with them, and gain deeper insights in the underlying model. The first initial partition is constructed as follows:

\begin{itemize}
	\item two of the inputs with identity function, $\mathit{TCR}$ and $\mathit{CD28}$, are kept in singleton blocks,
	\item variables with stable function $\mathit{true}$ belong to one block,
	\item variables with stable function $\mathit{false}$ belong to another block, and
	\item we define one more block containing $\mathit{TLR5}$ and all variables that belong to its equivalent class in the case of maximal reduction  %and its descendants \comal{what is a descendant? is it related to syntactical dependencies? what if a dependency is syntactic but not semantic?}, 
	i.e. $\{\mathit{TLR5}$, $\mathit{TICAM1}$, $\mathit{PIK3AP1}$, $\mathit{MyD88}$, $\mathit{IRAK1}$, $\mathit{TRAF6}$, $\mathit{TAB}$, $\mathit{TAK1}$, $\mathit{TAK1c}$, $\mathit{IRAK1}$, $\mathit{MAP2K3}$, $\mathit{MAP2K7}\}$-the blue chain of variables (Fig.~\ref{merged}~top). 
	%except the variable $\mathit{MAP2K3}$.
\end{itemize}

We call the reduced BN obtained by the first initial partition \emph{refined reduced model}~$\emph{I}$.

The second initial partition that we consider is similar to the first but one subtle differentiation: the variable $\mathit{MAP2K3}$ is kept in singleton block. The reduced BN that results from this initial partition, is called \emph{refined reduced model}~$\emph{II}$. The results of our study in this model is summarized in Table~\ref{resTCRTLR}~. We present the number of variables (\emph{Size}), the number of \emph{Attractors}, and the time needed for reduction (\emph{Reduction (s)}) and attractor identification (\emph{Analysis~(s)}).

\begin{table}[h]
    
	\centering
	\scalebox{0.95}{
		\begin{tabular}{c|c|c|c|c}
		    \hline
			\emph{Model} &\emph{Size} & \emph{Attractors} & \emph{Analysis (s)}& \emph{Reduction (s)} \\
			\hline
			\emph{Original} &128	& \multicolumn{2}{c|}{---\emph{Time Out}---}	&-	\\
			\hline
			\emph{Input-distinguished}&116	& \multicolumn{2}{c|}{---\emph{Time Out}---} &	2.058\\
			\hline
			\emph{Refined Reduced II}&98	&8	&9349.577 &2.096\\
			\hline
			\emph{Refined Reduced I}&97	&8	&1103.912 &1.833	\\
			\hline
			\emph{Maximal} &95	&2	& 29.336 &1.958 \\
		\end{tabular}
	}
	\caption{%Defining alternative initial partitions in Algorithm~\ref{algorithm} we obtain different reductions. 
		The results of the TCR-TLR merged BN for different reduced versions of the original model.}
	\label{resTCRTLR}
\end{table}

\paragraph{Interpretation} As we have seen before, attractor identification is intractable in the original and the ID-reduced BN whereas we can identify two attractors in the case of maximal reduction. Attractor identification in the refined reduced models is still feasible wherein we find more attractors than in the case of maximal reduction. We should highlight that reducing the TCR-TLR merged BN by just one variable decreases attractor identification time by several orders of magnitude (see \emph{Analysis time} in the case of \emph{Refined Reduced Models}). The computation of the BBE-reduced BN took less than $\mathit{3}$ seconds in all cases. To sum up, initializing Algorithm~\ref{algorithm} with alternative initial partitions derived from the results of the maximal and ID reduction enables us to explore richer behaviours of the original model.

\subsection{Mitogen-Activated Protein Kinase (MAPK) network}\label{sec:MAPK}

In this Section, we obtain a refined reduced model of the MAPK BN using results from the maximal and the ID reduced model.  The original MAPK BN~\cite{grieco2013integrative} consists of $\mathit{53}$ variables,   $\mathit{4}$ inputs and has $\mathit{40}$ attractors.  %that computed in $\mathit{16,501}$ seconds using Boolsim software.. 
%inputs: $\mathit{X_I=\{EGFR\_stimulus}$, $\mathit{FGFR3\_stimulus}$, $\mathit{TGFBR\_stimulus,DNA\_damage\}}$. 
We performed ID BBE-reduction and found the following equivalence classes: $\mathit{\{JNK,p38\} }$, $\mathit{\{SMAD}$, $\mathit{TAK1\}} $, $\mathit{\{ATF2, JUN}, \mathit{MAX}$, $\mathit{PPP2CA\}}$, $\mathit{\{ELK1}$, $\mathit{MSK\}}$, $\mathit{\{RSK, SPRY\}}$. The classes are displayed in the up part of Fig.~\ref{MAPK}~: Backward equivalent variables are represented with colored boxes and colored boxes that belong to the same class have the same background. Variables in white background belong to singleton classes. 

\paragraph{ID Reduction}
The ID reduced MAPK BN has $\mathit{46}$ variables and $\mathit{40}$ attractors. %that computed in $\mathit{14,768}$ seconds. 
Note that all attractors are preserved. This is a trivial consequence from the fact that (i) the number of attractors is the same, and (ii) the STG of the reduced BN is a subgraph of the STG of the original BN (isomorphism Lemma~\ref{lemma:isomorphism}). The BBE-reduction is consistent with~\cite{cardelli2017maximal}, where the authors transformed the BN to a system of ordinary differential equations, and reduced with backward differential equivalence. 
%ERODE reduced the BN in $\mathit{0,960} $ seconds.

\begin{figure}%[h]
	\centering
	\begin{tabular}{cc}
		\includegraphics[scale=0.30]{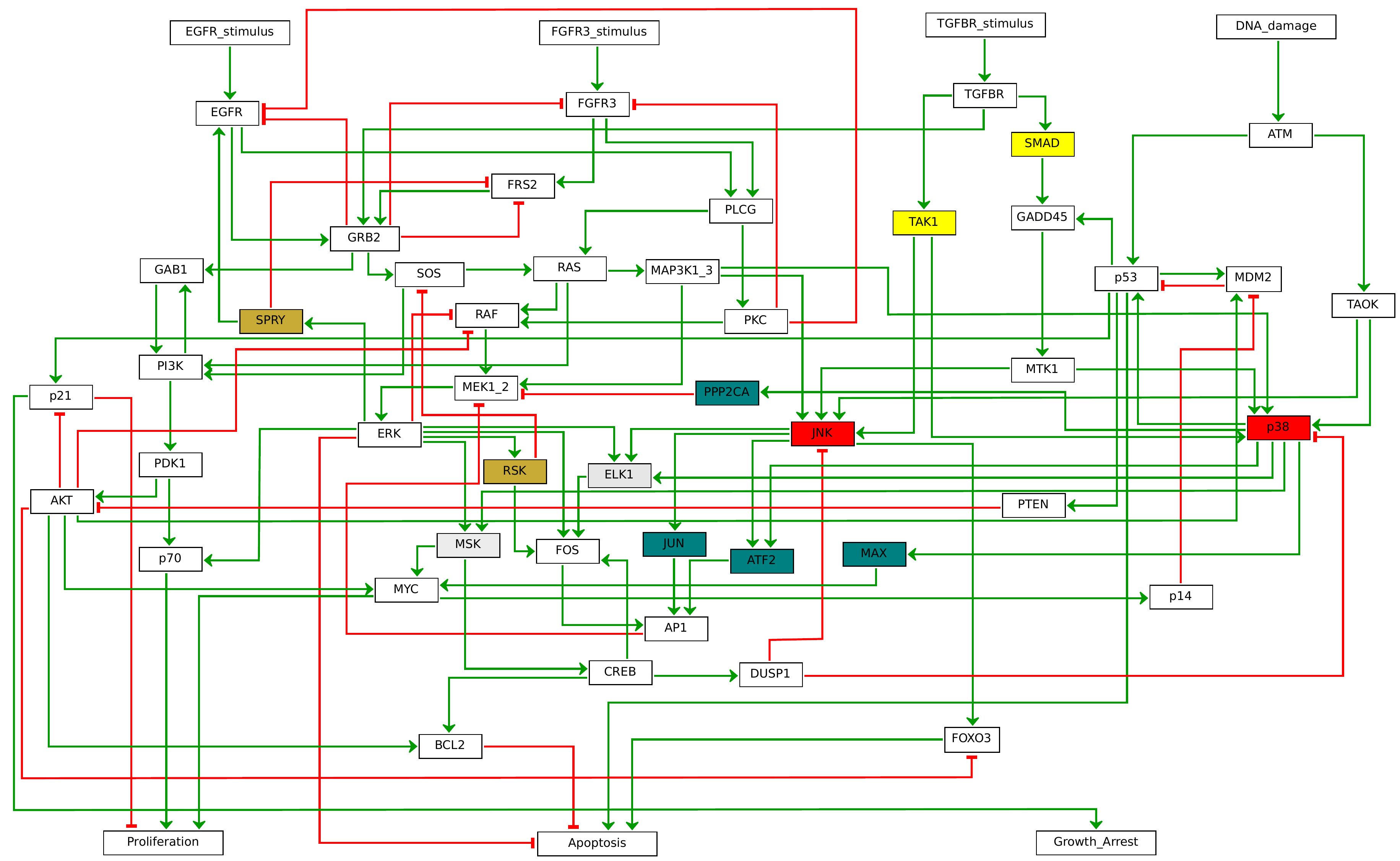}
		\\
		\hline
		\\
		\includegraphics[scale=0.30]{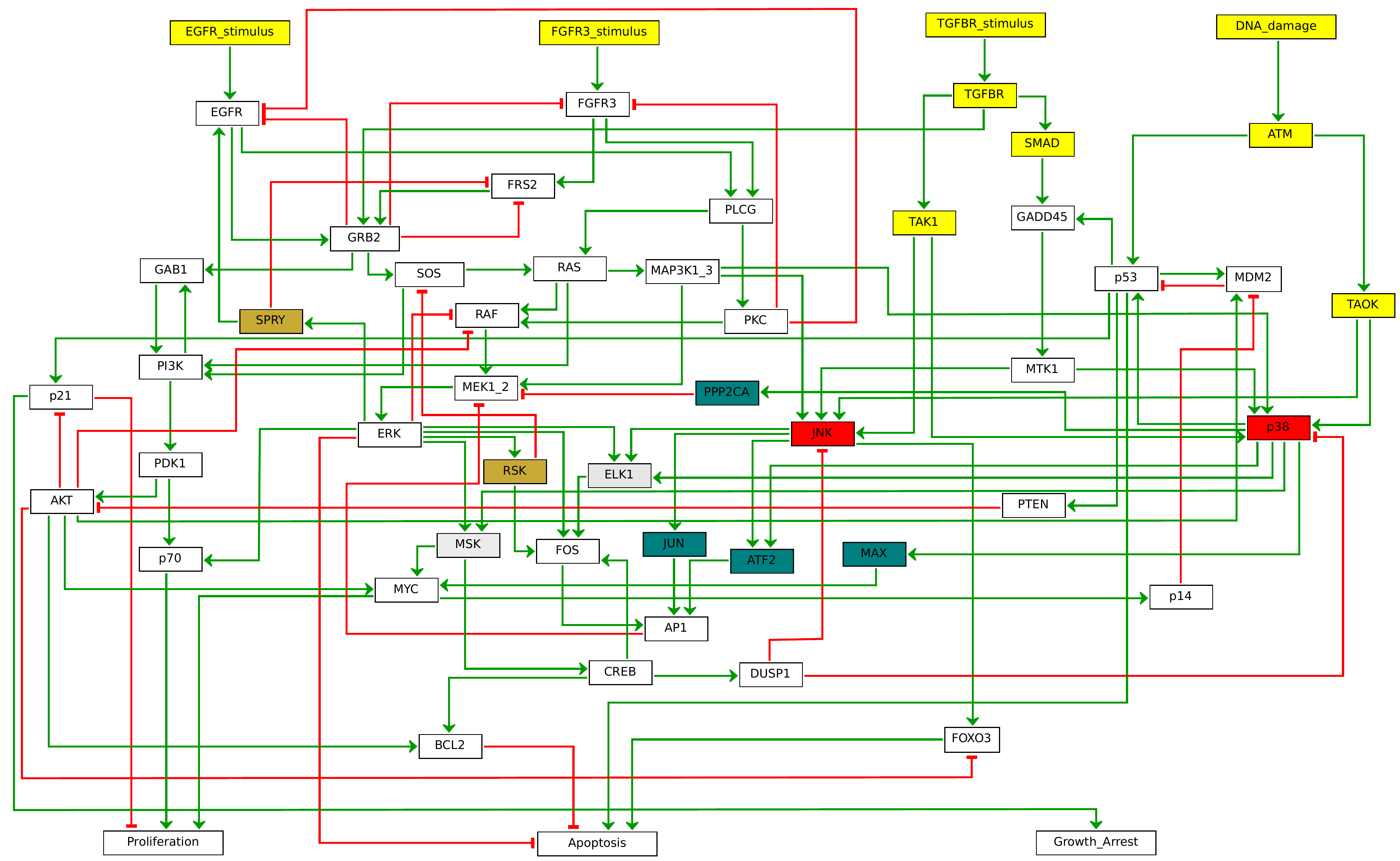}
	\end{tabular}
	\caption{Up: The MAPK BN with input-distinguishing BBE-variables having the same background color.
		Bottom: The MAPK BN with maximal-reduction BBE-variables having the same background color. Variables with white backround belong to singleton classes.}
	\label{MAPK}
\end{figure}

\paragraph{Maximal Reduction} The bottom part of Fig.~\ref{MAPK} displays the MAPK BN and its equivalence classes after performing the maximal reduction. Algorithm~\ref{algorithm} found the following equivalence classes: $\mathit{\{JNK,p38\} }$, $\mathit{\{SMAD}$, $\mathit{TGFBR}$, $\mathit{ATM} $,  $\mathit{TAOK} $, $\mathit{EGFR\_stimulus} $, $\mathit{FGFR3\_stimulus} $, $\mathit{TGFBR\_stimulus} $,  $\mathit{DNA\_damage} $, $\mathit{TAK1\}} $, $\mathit{\{ATF2}$, $\mathit{JUN}$, $\mathit{MAX}$, $\mathit{PPP2CA\}}$, $\mathit{\{ELK1}$, $\mathit{MSK\}}$, $\mathit{\{RSK, SPRY\}}$. BoolSim computes $\mathit{17}$ attractors in %$\mathit{2,418}$ seconds in 
this case which means that the number of attractors is not preserved. In contrast with \cite{naldi2011dynamically}, the preserved attractors are pure in the original network in the sense that the isomorphism of Lemma~\ref{lemma:isomorphism} translates the attractors of the reduced to the original BN.

\paragraph{Refined Reduction}
Based on observations gained from the maximal and the ID reduction, we specify the following partition: $\{\mathit{EGFR\_stimulus}\}$, $\{\mathit{FGFR3}$ $\mathit{\_stimulus}\}$, $\{\mathit{TGFBR\_stimulus}$,
$\mathit{TGFBR} $, $\mathit{TAK1} $, $\mathit{SMAD}\}$, $\{\mathit{DNA\_damage}$, $\mathit{ATM} $,  $\mathit{TAOK}\}$, and one block containing all the remaining variables. In other words, we keep 2 of the inputs ($\{\mathit{EGFR\_stimulus}\}$, $\{\mathit{FGFR3\_stimulus}\}$) still in singleton sets, while we define two more blocks with each one containing one input and the input's BBE-variables found in the maximal reduction. %
We expect that the reduced BN, which now contains $\mathit{42}$ variables ($\mathit{79,25\%}$ of the original size), preserves more properties. Indeed, the refined reduced BN has all $\mathit{40}$ attractors of the original BN. The results of our study in this model is summarized in the following Table~\ref{MAPKres}~. We present the number of variables (\emph{Size}), the number of \emph{Attractors}, and the time needed for reduction (\emph{Reduction (s)}) and attractor identification (\emph{Analysis~(s)}).

\begin{table}[h]
	\centering
	\scalebox{0.95}{
		\begin{tabular}{c|c|c|c|c}
		    \hline
			\emph{Model} &\emph{Size} & \emph{Attractors} & \emph{Analysis (s)} & \emph{Reduction (s)} \\
			\hline
			\emph{Original} &53	& 40&16.501	&	-	\\
			\hline
			\emph{Input-distinguished}&	46& 40&	14.480&	0.848\\
			\hline
			\emph{Refined Reduced} &42& 40&12.115&1.202\\
			\hline
			\emph{Maximal}&39&17&2.471&1.018 \\
		\end{tabular}
	}
	\caption{%Defining alternative initial partitions to Algorithm \ref{algorithm} we obtain different reductions. 
		The results of the MAPK BN for different reduction versions of the original model.}
	\label{MAPKres}
\end{table}

\subsection{ T cell granular lymphocyte (T-LGL) leukemia BN}\label{sec:TLGL}

T-LGL BN was originally introduced in \cite{zhang2008network} and refers to the disease T-LGL leukemia which features a clonal expansion of antigen-primed, competent, cytotoxic T lymphocytes (CTL). The T-LGL BN is a signalling pathway, constructed empirically through extensive literature review, and determines the survival of CTL. The ID reduction and the maximal reduction are depicted in the top part and the bottom part of Fig.~\ref{TLGL} respectively. The original BN consists of 60 variables, and has 264 attractors. %which are computed in 123.431 seconds.

\begin{figure}%[h]
	\centering
	\begin{tabular}{c}
		\includegraphics[scale=0.30]{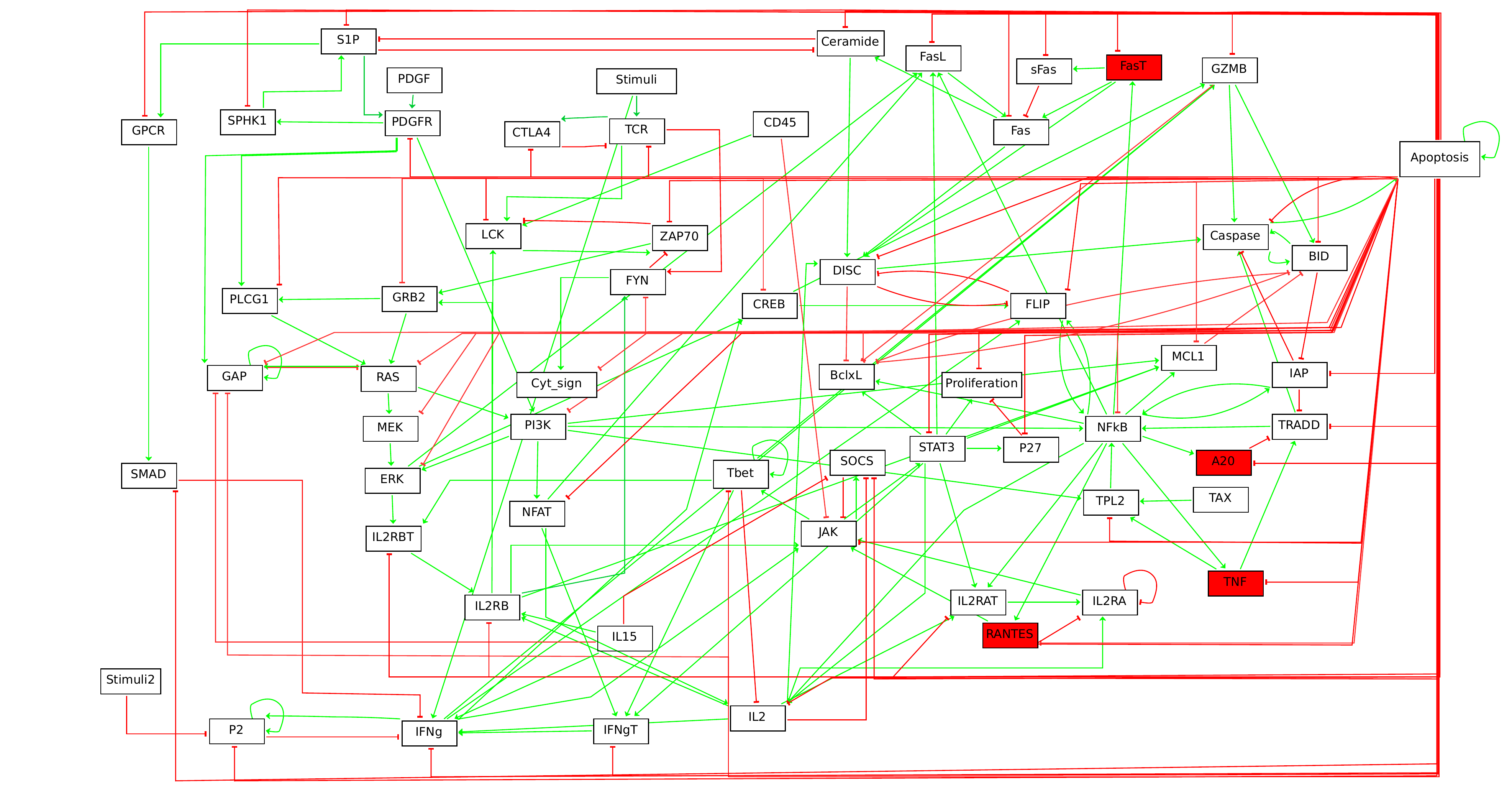}
		\\
		\hline
		\\
		\includegraphics[scale=0.30]{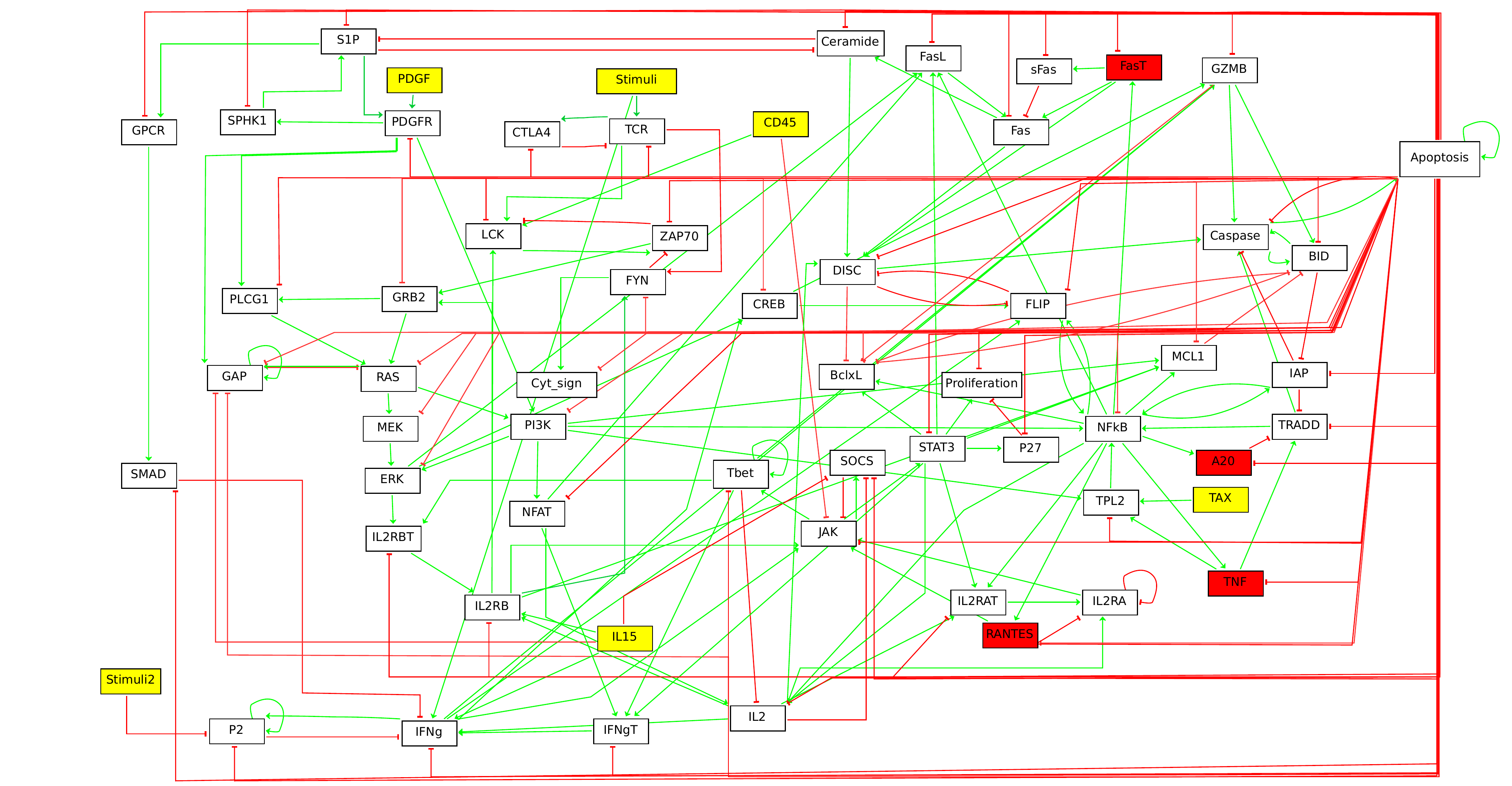}
	\end{tabular}
	\caption{Up: The T-LGL BN with input-distinguishing BBE-variables having the same background color.
		Bottom: The T-LGL BN with maximal-reduction BBE-variables having the same background color. Variables with white backround belong to singleton classes.}
	\label{TLGL}
\end{figure}

In the case of ID reduction, the variables $\mathit{FasT}$, $\mathit{A20}$, $\mathit{TNF}$ and $\mathit{RANTES}$ are BBE-equivalent so we can collapse them into a single variable. The ID reduced BN has 57 variables, and 264 attractors. Since the number of attractors is the same, and the STG of the reduced BN is a subgraph of the STG of the original BN, the asymptotic dynamics are preserved. The bottom part of Fig.~\ref{TLGL} refers to the maximal BBE. In this case, we have two equivalence classes: the one found in ID BBE, and one consisted of all the input variables. %The overall time needed for the reduction and the attractor identification of the reduced BN is 87,282 seconds. 
On the other hand, the maximal reduced BN has 52 variables and 6 attractors. This means that some attractors are lost.

In \cite{zhang2008network}, the authors considered the asynchronous schema and presented a variable specified analysis. Specifically, their analysis determined which variables are sufficient to induce all of known signalling abnormalities in leukemic T-LGL, which variables are important for the survival of leukemic T-LGL, and which variables are constantly active in leukemic T-LGL. Notably, permanent activation of the variables $\mathit{IL-15}$ and $\mathit{PDGF}$ is sufficient to produce all of the known deregulations and signalling abnormalities. For this reason, we consider as reasonable initial partition one wherein $\mathit{IL-15}$ and $\mathit{PDGF}$ belong to the same block, other input variables belong to singleton blocks, and non-input variables belong to one block. In this case, the refined reduced BN has 56 variables and $120$ attractors. %computed $55.049$ seconds. 
In contrast with the maximal reduced and the original BN which have 264 attractors, the BN reduced with this reasonable initial partition discards 144 attractors which are irrelevant for this kind of analysis. The results of our study on this model is summarized in the following Table~\ref{TLGLres}~. We present the number of variables (\emph{Size}), the number of \emph{Attractors}, and the time needed for reduction (\emph{Reduction (s)}) and attractor identification (\emph{Analysis~(s)}).

\begin{table}[h]
	\centering
	\scalebox{0.95}{
		\begin{tabular}{c|c|c|c|c}
			\emph{Model} &\emph{Size} & \emph{Attractors} & \emph{Analysis (s)}& \emph{Reduction (s)} \\
			\hline
			\emph{Original} &60	& 264&123.431&-		\\
			\hline
			\emph{Input-distinguished}&57	& 264	&85.999&0.843		\\
			\hline
			\emph{Refined Reduced} &56	&120	&55.049 &1.816	\\
			\hline
			\emph{Maximal} &52	&6&2.489 &0.999	 \\
		\end{tabular}
	}
	\caption{%Defining alternative initial partitions to Algorithm~\ref{algorithm} we obtain different reductions. 
		The results of the T-LGL BN for different reduction versions of the original model.}
	\label{TLGLres}
\end{table}

%%\begin{table}[H]
%	\centering
%	\scalebox{0.6}{
%		\begin{tabular}{c|cccc|ccc}
%			\hline
%			\emph{Model} &   \multicolumn{4}{c|}{\emph{Original model}}     &  \multicolumn{3}{c}{\emph{maximal Reduced model}} \\
%			\hline
%			& \emph{Size}& \emph{Attractors} & \emph{Analysis (s)} & \emph{Reduction (s)}&\emph{Size}& \emph{Attractors} & \emph{Analysis (s)}\\
%			%\hline
%TCRsig&40& 8 & 0,148s & 0,673s&31  &8 & 0,109s\\
%			\hline
%			MAPK Network& 53 &40  & 16.501  & 1.471 &39 &17 &2.418 \\
%			\hline
%			T-LGL 		&60  &264 & 123.431 & 1.422  &52 &6 &2.066  \\
%			\hline
%			TCR-TLR merged &128 &\multicolumn{2}{c}{---\emph{Time Out}---} & 2.615 & 95 &2 & 28.679 \\
%		\end{tabular}
%	}
%	\caption{The number of attractors is not preserved by the maximal reduction. \comal{As before, I believe the important thing is to note that some atractors are not preserved, not that the number itself is not preserved.} The computational costs is much better in the last two cases and the state space of the reduced model is smaller by several orders of magnitude. The analysis time refers to the computation time of attractors.\comal{I wonder if we should present reduction+analysis time as a column to make the life of readers easier.}}
%	\label{rt1}
%\end{table}

%\comal{To be discussed: move this to the beginning of the section or present it in the individual subsections? The first option is good since it would appear close to Table 3 where one easily see whether the refined case is between input and max.}

\section{Speed-ups on STG generation in BBE-reduced models} \label{app:speedups}

\paragraph{Hypothesis}
%In this Section, we investigate if BBE-reduction can be applied to reduce the computation time of the STG generation and to facilitate STG visualization. 
We hope to drastically reduce the time needed for STG generation. Furthermore, we claim that our technique may be utilized for STG visualization since reducing a BN only by one variable results in reducing its corresponding STG by $\mathit{50\%}$.

%We generated the STGs of this section with PyBoolNet~\cite{pyboolnet}, which is a python package for generation, analysis and visualisation of BNs. 

\paragraph{Configuration}
We reduced the original BN with both ID and maximal BBE-reduction. We observed that PyBoolNet failed to generate the STG of BNs that have more than 25 variables. Hence, we restricted our experiments to all BNs with less variables. PyBoolNet generates the STG within several minutes for BNs between 21 and 25 variables, and within a minute for BNs with up to 20 variables. We did not consider BNs with less than 9 variables since the generation and visualization of the full STG is feasible and computationally costless. We present the results in Table~\ref{rt2}:

\begin{table}[h]
	\centering
	\scalebox{0.75}{
		\begin{tabular}{c|cc|ccc|ccc}
			\hline
			\emph{Model} &   \multicolumn{2}{c|}{\emph{Original model}}     &  \multicolumn{3}{c|}{\emph{Input-distinguished Reduced model}} &  \multicolumn{3}{c}{\emph{Maximal Reduced model}} \\
			\hline
			& \emph{Size}& \emph{STG generation(s)} & \emph{Reduction (s)}&\emph{Size}& \emph{STG generation(s)}& \emph{Reduction (s)}&\emph{Size}& \emph{STG generation(s)} \\
			\hline
			%Cell-Fate Decision (i) 
			B7~%\cite{calzone2010mathematical} 
			&33 &\emph{out of memory}       &0.585	&27	&\emph{out of memory}   &0.608	&25 & \emph{out of memory}\\
			\hline
			%Cell-Fate Decision (ii) 
			B9 %\cite{calzone2010mathematical} 
			&28 &\emph{out of memory}       &0.449	&25	&\emph{out of memory}   &0,416	&20 & 52.8 \\
			\hline
			%Drosophila Wg 
			B10 %\cite{mbodj2013logical} 		
			&26 &\emph{out of memory}		&0.227  &23 &457  &0.145	&4 & 0.006\\
			\hline
			%Drosophila Hh 
			B11 %\cite{mbodj2013logical}  		
			&24 &984       &0.243	&23	&475  &0.207	&9 & 0,280\\
			\hline
			%Drosophila Spaltze 
			B12 %\cite{mbodj2013logical}  
			&24 &987       &0.349	&21	&102  &0.121	&4 &0.050 \\
			\hline
			%Drosophila FGF
			B13 %\cite{mbodj2013logical}  	
			&23 &455       &0.302	&22	&226  &0.176	&8 &0.164 \\
			\hline
			%ERBB
			B14 %\cite{sahin2009modeling}				
			&20 &55.6  	&0.497  &15	&2.11  &0.408 	&13 & 0.302 \\
			\hline
			%Drosophila VEGF
			B15 %\cite{mbodj2013logical}		
			&18 &11.6       &0.209	&18	&11.6   &0,182	&8 &0.007 \\
			\hline
			%BY Cell cycle
			B16 %\cite{irons2009logical}
			&18 &14.300 & \multicolumn{3}{c|}{---\emph{NO INPUTS}---} &0.449 &17 &6.760\\
			\hline
			%D Cell cycle
			B17 %\cite{faure2009logical}
			&14 &0,867      &0.267	&14	&0.867   &0.389	&12 &0.169 \\
			\hline
			%Drosophila Toll sig 
			B18 %\cite{mbodj2013logical}
			&11 &0.072       &0.327	&10	&0.064   &0.214	&9 &0.065 \\
			\hline
			%FY Cell Cycle
			B19 %\cite{faure2009logical}  
			&10 &0.044       &0.228	&9	&0.016   &0.303	&9 &0.016 \\
			\hline
			%DV
			B20 %\cite{gonzalez2006dynamical} 			
			&10 &0.172   	&0.283  &8 	&0.044  &0.202	&8 &0.044 \\
			\hline
			%Cell Division g2b
			B21 %\cite{sanchez2017modeling} 
			&9 &0.015       & \multicolumn{3}{c|}{---\emph{NO INPUTS}---} &0.279	&7 &0.005 \\
			\hline
			%BY Cell cycle
			B22 %\cite{faure2009logical} 
			&9 &0.025 & \multicolumn{3}{c|}{---\emph{NO INPUTS}---} &0.237 &7 &0.003\\
			\hline
			%miR-9, neurogenesis
			%B23 \cite{coolen2012mir} &6 &       &	 \multicolumn{3}{c|}{---\emph{NO INPUTS}---}   &	&2 & \\
			%\hline
			%Cell Division g2a
			%B24 \cite{sanchez2017modeling} &5 &  0.157     & \multicolumn{3}{c|}{---\emph{NO INPUTS}---}   &	&1 & \\
			M9  %\cite{mbodj2016qualitative} 
			&57 &\emph{out of memory} & 0.791 &57 &\emph{out of memory} &0.260 &11 &0.233 \\ \hline
			M16 %\cite{sanchez2016primary} 
			&37 &\emph{out of memory} & 0.907 &37 &\emph{out of memory} &0.454 &14 &1.360 \\\hline
			M17 %\cite{naldi2010diversity} 
			&36 &\emph{out of memory} &0.413 &35 &\emph{out of memory} &0.516 &21 &136 \\\hline
			M20 %\cite{mbodj2013logical} 
			&34 &\emph{out of memory} &0.364 &32 &\emph{out of memory} &0.383 &15 &2.68 \\\hline
			M21 %\cite{mombach2014modelling} 
			&30 &\emph{out of memory} &0.421 &30 &\emph{out of memory} &0.238 &14 &1.37 \\\hline
			M22 %\cite{faure2014discrete} 
			&24 &1212 &0.251 &23 &1043 &0.219 &12 &0.172 \\\hline
			M23 %\cite{mendoza2006network} 
			&21 &130 &0.273 &21 &130 &0.326 &18 &14.6 \\\hline
			M24 %\cite{mbodj2013logical} 
			&19 &31 &0.109 &19 &31 &0.153 &7 &0.463 \\\hline
			M25 %\cite{sanchez2016primary} 
			&19 &28.3 &0.210 &19 &28.3 &0.243 &12 &0.609 \\\hline
			M26 %\cite{sanchez2002segmenting} 
			&19 &30.1 &0.249 &18 &13.9 &0.356 &17 &6.320 \\\hline
			M27 %\cite{mbodj2013logical} 
			&18 &14.2 &0.161 &18 &14.2 &0.194 &10 &0.260 \\\hline
			M28 %\cite{faure2009modular} 
			&16 &3.34 &0.189 &16 &3.34 &0.266 &13 &1.54 \\\hline
			M29 %\cite{mbodj2013logical} 
			&16 &3.15 &0.101 &16 &3.15 &0.096 &5 &0.028 \\\hline
			M30 %\cite{sanchez2018logical} 
			&15 &1.59 &0.187 &15 &1.59 &0.235 &13 &0.303 \\\hline
			M31 %\cite{flobak2015discovery} 
			&14 &0.883 &0.178 &14 &0.883 &0.203 &13 &0.444 \\\hline
			M32 %\cite{faure2009modular} 
			&12 &0.156 &0.168 &12 &0.156 &0.137 &8 &0.010 \\\hline
			M33 %\cite{sanchez2018logical} 
			&10 &0.032 &0.098 &10 &0.032 &0.044 &8 &0.007 \\\hline
		\end{tabular}
	}
	\caption{Time needed for model reduction and STG generation of the original and the reduced BN. The running times are coming from one run and the computation of the BBE-reduced BNs take no more than 1 second in the worst cases.} 
	\label{rt2}
\end{table}

\paragraph{Results} PyBoolNet failed to generate the STG of the original B9~\cite{calzone2010mathematical}. This was done within a minute after applying maximal BBE-reduction. For BNs between 20 and 25 variables, our method drastically decreased the STG generation time: The STG of the ID BN needs on average $\mathit{25\%}$ of the time for the generation of the full STG, and the maximal BN needs less than $\mathit{1\%}$ of the time needed for the generation of the full STG. For BNs with less than 20 variables the reduction may be computationally effective in several cases (see B15 and B16 in Table~\ref{rt2}).% For BNs with less than 15 variables, BBE-reduction is not computationally effective but it is efficient for the visualization of the STG. We display an example of this case in the~\ref{appB}~. 

\paragraph{Interpretation} We should note that our method (i) may render the analysis of large BN models tractable in many cases (like B9, M9, and M21) and (ii) facilitates STG visualization. BBE-reduction constitutes a useful method for generating pure segments of the original state space. According to the isomorphism Lemma~\ref{lemma:isomorphism}, the STG of the reduced BN constitutes a subgraph of the STG of the original BN resulting from it after the collapse of a BBE-class into one variable component. In other words, our reduction method provides a pure image of the state space of the original BN. We utilize these results in Section~\ref{casestudies} wherein we conduct an attractor based analysis.

\end{document}